\documentclass{ejm}

\pdfoutput=1 

\ifprodtf \else
  \iffontfound
    \IfFileExists{upmath.sty}
      {\typeout{^^JFound AMS Euler Roman fonts on the system,
                   using the 'upmath' package.^^J}%
       \usepackage{upmath}}
      {\typeout{^^JFound AMS Euler Roman fonts on the system, but you
                   don't seem to have the}%
       \typeout{'upmath' package installed. EJM.cls can take advantage
                 of these fonts,^^Jif you use 'upmath' package.^^J}%
       \providecommand\mathup[1]{##1}%
       \providecommand\mathbup[1]{##1}%
      }
  \else
    \providecommand\mathup[1]{#1}%
    \providecommand\mathbup[1]{#1}%
  \fi
\fi

\ifprodtf \else
  \iffontfound
    \IfFileExists{amssymb.sty}
      {\typeout{^^JFound AMS Symbol fonts on the system, using the
                'amssymb' package.^^J}%
       \usepackage{amssymb}%
         \let\leq=\leqslant
       \let\ge=\geqslant  \let\geq=\geqslant
      }{}
  \fi
\fi

\ifprodtf \else
  \IfFileExists{amsbsy.sty}
    {\typeout{^^JFound the 'amsbsy' package on the system, using it.^^J}%
     \usepackage{amsbsy}}
    {\providecommand\boldsymbol[1]{\mbox{\boldmath $##1$}}}
\fi

\newsavebox{\astrutbox}
\sbox{\astrutbox}{\rule[-5pt]{0pt}{20pt}}

\newtheorem{theorem}{Theorem}[section]
\newdefinition{definition}[theorem]{Definition}

\usepackage{color,array}
\usepackage{epstopdf}
\usepackage{paralist}
\usepackage{subfigure}

\usepackage[nice]{units} 

\usepackage{dgleich-setup}

\usepackage[utf8]{inputenc}

\usepackage{amsfonts,amsmath,amssymb}

\usepackage{array}

\providecommand{\mA}{\ensuremath{\textbf{A}}}

\providecommand{\mD}{\ensuremath{\textbf{D}}}

\providecommand{\mI}{\ensuremath{\textbf{I}}}

\providecommand{\mM}{\ensuremath{\textbf{M}}}

\providecommand{\mP}{\ensuremath{\textbf{P}}}
\providecommand{\mQ}{\ensuremath{\textbf{Q}}}

\providecommand{\vb}{\ensuremath{\textbf{b}}}

\providecommand{\ve}{\ensuremath{\textbf{e}}}

\providecommand{\vg}{\ensuremath{\textbf{g}}}

\providecommand{\vr}{\ensuremath{\textbf{r}}}
\providecommand{\vs}{\ensuremath{\textbf{s}}}

\providecommand{\vv}{\ensuremath{\textbf{v}}}

\providecommand{\vx}{\ensuremath{\textbf{x}}}
\providecommand{\vy}{\ensuremath{\textbf{y}}}
\providecommand{\vz}{\ensuremath{\textbf{z}}}

\RequirePackage{graphicx}
\RequirePackage{tabularx}
\RequirePackage{booktabs}
\RequirePackage{ragged2e}
\RequirePackage{xcolor}
\RequirePackage{textcomp} 
\RequirePackage{xspace}
\RequirePackage{microtype}
\RequirePackage{listings}

\colorlet{TufteRed}{red!80!black}
\definecolor{halfgray}{gray}{0.55}

\definecolor{subtleblue}     {rgb}{0.02,0.04,0.48}
\definecolor{subtlered}      {rgb}{0.65,0.04,0.07} 
\definecolor{subtlegreen}    {rgb}{0.06,0.44,0.08}
\definecolor{subtledarkblue} {rgb}{0,.1,.6}
\definecolor{lightsubtleblue}{rgb}{0,.4,.6}
\definecolor{ecru}           {rgb}{1.0,.98823,.95686}   

\definecolor{stanfordred}      {rgb}{0.6431,0.000,0.1137} 

\definecolor{stanfordsandstone}{rgb}{0.9059,0.8196,0.6039}
\definecolor{stanfordblue}     {rgb}{0.1451,0.5176,0.7333}
\definecolor{stanfordgreen}    {rgb}{0.1608,0.3961,0.2863} 
\definecolor{stanforddarkgray} {rgb}{0.2627,0.2902,0.2667} 
\definecolor{stanfordlightgray}{rgb}{0.8392,0.8667,0.8275} 
\definecolor{stanforddarkgreen}{rgb}{0.2353,0.2118,0.1373}
\definecolor{stanforddeepred}  {rgb}{0.6510,0.2275,0.0000}
\definecolor{stanfordneutralkhaki}{rgb}{0.5686,    0.5333,    0.4510}
\definecolor{stanfordbrightgreen}{rgb}{0.0039 ,   0.5137  ,  0.3725}
\definecolor{stanfordbrightblue}{rgb}{0.1451,0.5176, 0.7333} 
\definecolor{stanfordbrightseagreen}{rgb}{0.0000,0.5020,0.5529} 
\definecolor{stanfordbrightyellow}{rgb}{0.9412,0.6863,0.0000} 
\definecolor{stanfordbrightwine}{rgb}{0.2353,0.0667,0.0275}

\newcommand{\PreserveBackslash}[1]{\let\temp=\\#1\let\\=\temp}
\newcolumntype{L}[1]{>{\PreserveBackslash\RaggedRight}m{#1}}
\newcolumntype{M}[1]{>{\PreserveBackslash\RaggedRight}p{#1}}
\newcolumntype{R}[1]{>{\PreserveBackslash\RaggedLeft}m{#1}}
\newcolumntype{S}[1]{>{\PreserveBackslash\RaggedLeft}p{#1}}
\newcolumntype{Z}[1]{>{\PreserveBackslash\Centering}m{#1}}
\newcolumntype{A}[1]{>{\PreserveBackslash\Centering}p{#1}}

\newcolumntype{U}{>{\setlength{\RaggedRightParindent}{0pt}\RaggedRight\arraybackslash\noindent}X}
\newcolumntype{V}{>{\RaggedLeft\arraybackslash}X}
\newcolumntype{W}{>{\Centering\arraybackslash}X}

\graphicspath{{.}{./figures/}}

\lstnewenvironment{pseudocode}[1][]{%
  \lstset{language=python,#1
  }%
}{}

\lstdefinelanguage{matlabfloz}{%
  alsoletter={...},%
  morekeywords={
  break,case,catch,continue,elseif,else,end,for,function,global,%
  if,otherwise,persistent,return,switch,try,while,...,ones,zeros,eye},%
  comment=[l]\%,
  morecomment=[l]...,
  morestring=[m]',
}[keywords,comments,strings]%

\lstnewenvironment{dispmcode}{%
    \lstset{language=matlabfloz,%
      numbers=none,frame=none,%
      keywordstyle=\color{subtleblue},
      commentstyle=\color{thegreen},
      stringstyle=\color{subtleed}
    }%
  }%
  {}

\lstnewenvironment{dispmresults}{%
    \lstset{language=matlabfloz,%
      numbers=none,frame=none,backgroundcolor=\color{ecru},%
      keywordstyle=\color{theblue},
      commentstyle=\color{subtlegreen},
      stringstyle=\color{subtlered}
    }%
  }%
  {}

\lstnewenvironment{numberedmcode}{%
    \lstset{language=matlabfloz,%
      numbers=left,frame=none,%
      keywordstyle=\color{subtleblue},
      commentstyle=\color{subtlegreen},
      stringstyle=\color{subtlered}
    }%
  }%
  {}

\newcommand{\inv}{^{-1}}\newcommand{\eps}{\varepsilon}
\newcommand{\epsmn}{\eps_{\min}}
\newcommand{\epsmx}{\eps_{\max}}
\newcommand{\epscu}{\eps_{\text{cur}}}

\newcommand{\hvx}{\hat{\vx}}
\newcommand{\hvy}{\hat{\vy}}

\newcommand{\vvk}[2]{\ensuremath{\textbf{#1}}^{(#2)}}

\newcommand{\vol}{\text{vol}}

\newcommand{\pprg}{\texttt{ppr-grid}\xspace}
\newcommand{\ppra}{\texttt{ppr-path}\xspace}

\newcommand{\ppath}{\ppra}
\newcommand{\pgrid}{\pprg}
\newcommand{\pgrow}{\texttt{ppr-grow}\xspace}

\newfont{\mycrnotice}{ptmr8t at 7pt}
\newfont{\myconfname}{ptmri8t at 7pt}
\usepackage{graphicx}

\title[]{Seeded PageRank Solution Paths}

\author[D. F. Gleich and K. Kloster]{%
D.~F.~Gleich$\,^1$, \and K.~Kloster$\,^2$
}

\affiliation{%
  $^1\,$Department of Computer Science, Purdue University, West Lafayette IN, USA\\
    email\textup{\nocorr: \texttt{dgleich@purdue.edu}}\\
  $^2\,$Department of Mathematics, Purdue University, West Lafayette IN, USA\\
    email\textup{\nocorr: \texttt{kkloste@purdue.edu}}\\
}

\begin{document}

\label{firstpage}
\maketitle

\begin{abstract}%
We study the behavior of network diffusions based on the PageRank random walk from a set of seed nodes. These diffusions are known to reveal small, localized clusters (or communities) and also large macro-scale clusters by varying a parameter that has a dual-interpretation as an accuracy bound and as a regularization level. We propose a new method that quickly approximates the result of the diffusion for all values of this parameter.
Our method efficiently generates an approximate \emph{solution path} or \emph{regularization path} associated with a PageRank diffusion, and it reveals cluster structures at multiple size-scales between small and large. We formally prove a runtime bound on this method that is independent of the size of the network, and we investigate multiple optimizations to our method that can be more practical in some settings.
We demonstrate that these methods identify refined clustering structure on a number of real-world networks with up to 2 billion edges. 

\end{abstract}



\section{Introduction}
\label{sec:intro}

Networks describing complex technological and social systems display many types of structure. One of the most important types of structure is clustering because it reveals the modules of technological systems and communities within social systems. A tremendous number of methods and objectives have been proposed for this task (survey articles include refs.~\cite{Schaeffer-2007-clustering,Xie-2013-overlapping}). The vast majority of these methods seek large regions of the graph that display evidence of local structure. For the case of modularity clustering, methods seek statistically anomalous regions; for the case of conductance clustering, methods seek 
dense regions that are weakly connected to the rest of the graph.
All of the objective functions designed for these clustering approaches implicitly or explicitly navigate a trade-off between cluster size and the underlying clustering signal. For example, large sets tend to be more anomalous than small sets. 
Note that these trade-offs are essential to multi-objective optimization, and the choices in the majority of methods are natural.
Nevertheless, directly optimizing the objective makes it difficult to study these structures as they vary in size from small to large because of these implicit or explicit biases. This intermediate regime represents the meso-scale structure of the network. 

In this manuscript, we seek to study structures in this meso-scale regime 
by analyzing the behavior of seeded graph diffusions.
Seeded graph diffusions model the behavior of a quantity of ``dye" that is continuously injected at a small set of vertices called the \emph{seeds} and distributed along the edges of the graph. These seeded diffusions can reveal multi-scale features of a graph through their dynamics. The class we study can be represented in terms of a column-stochastic distribution operator $\mP$: 
\[ \vx = \textstyle \sum_{k=0}^\infty \gamma_k \mP \vs \]
where $\gamma_k$ are a set of diffusion coefficients that reflect the behavior of the dye $k$ steps away from the seed, and $\vs$ is a sparse, stochastic vector representing the seed nodes.
More specifically, we study the PageRank diffusions 
\[ \vx = \textstyle \sum_{k=0}^\infty (1-\alpha) \alpha^k \mP \vs. \]
The PageRank diffusion is equivalent to the stationary distribution of a random walk that (i) with probability $\alpha$, follows an edge in the graph and (ii) with probability $(1-\alpha)$ jumps back to a seed vertex (see Section~\ref{sec:prelims} more detail on this connection).


PageRank itself has been used for a broad range of applications including data mining, machine learning, biology, chemistry, and neuroscience; see our recent survey~\cite{Gleich-2015-prbeyond}.
Among all the uses of PageRank, the \emph{seeded variation} is frequently used to localize the PageRank vector within a subset of the network; this is also known as \emph{personalized PageRank} due to its origins on the web, or \emph{localized PageRank} because of its behavior.(We will use these terms: seeded PageRank, personalized PageRank, and localized PageRank interchangeably and use the standard acronym PPR to refer to them.)  Perhaps the most important justification for this use is presented in~\cite{andersen2006-local}, where the authors
determined a relationship between seeded PageRank vectors and low-conductance sets that allowed them to create a type of graph partitioning method that does not need to see the entire graph. Their PageRank-based clustering method, called the \emph{push method}, has been used for a number of important insights into communities in large social and information networks~\cite{jeub2015locally,Leskovec-2009-community-structure}.


Our focus is a novel application of this push method for meso-scale structural analysis of networks. Push, which we'll describe formally in Section~\ref{sec:prelim-push}, depends on an accuracy parameter $\eps$. As we vary $\eps$, the result of the push method for approximating the PageRank diffusion reveals different structures of the network. We illustrate three PageRank vectors as we vary $\eps$ for Newman's network science collaboration graph~\cite{Newman-2006-eigenvectors} in Figure~\ref{fig:netsci}. There, we see that the solution vectors for PageRank
that result from push have only a few non-zeros for large values of $\eps$. 
(Aside: There is a subtle inaccuracy in this statement. As we shall see shortly, we actually are describing degree normalized PageRank values. This difference does not affect the non-zero components or the intuition behind the discussion.)
This is interesting because an accurate PageRank vector is mathematically non-zero everywhere in the graph. Push, with large values of $\eps$, then produces sparse approximations to the PageRank vector. This connection is formal, and the parameter $\eps$ has a dual interpretation as a sparsity regularization parameter~\cite{Gleich-2014-alg-anti-diff} (reviewed in Section~\ref{sec:regularization}).

\begin{figure*}[t]
	\centering
	\subfigure[$\eps = 10^{-2}$]{\includegraphics[width=0.3\linewidth]{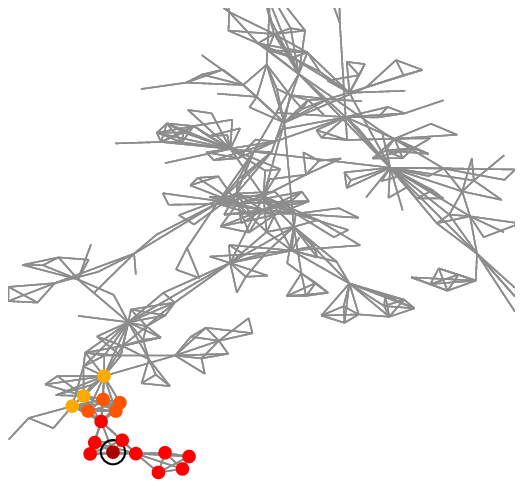}}
	\subfigure[$\eps = 10^{-3}$]{\includegraphics[width=0.3\linewidth]{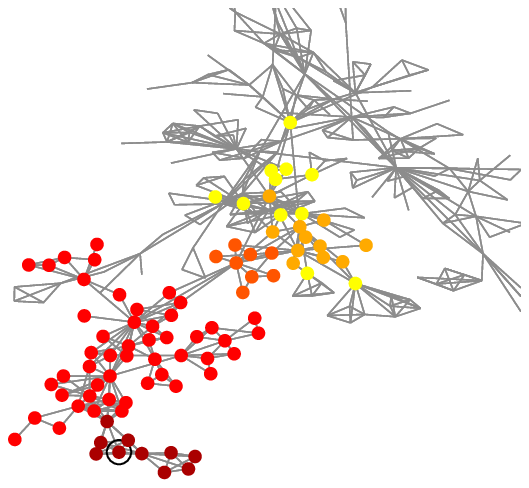}}
	\subfigure[$\eps = 10^{-4}$]{\includegraphics[width=0.3\linewidth]{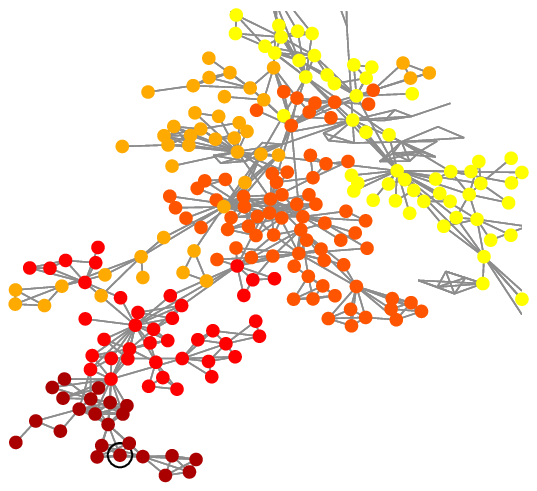}}
	\caption{Nodes colored by their degree-normalized PageRank values as $\eps$ varies: dark red is large, yellow is small.
		The hidden nodes are mathematically zero. As $\eps$ decreases, more nodes become non-zero.
	}
	\label{fig:netsci}
\end{figure*}

The solution path or regularization path for a parameter is the set of trajectories that the components of the solution trace out as the parameter varies~\cite{Efron-2004-lars}.  We present new algorithms based on the \emph{push procedure} that allow us to approximate the solution path trajectories as a function of $\eps$. We use our solution path approximation to explore the properties of graphs at many size-scales in Section~\ref{sec:paths}. In our technical description, we show that the solution path remains localized in the graph (Theorem~\ref{thm:ppra}). Experiments show that it runs on real-world networks with millions of nodes in less than a second (Section~\ref{sec:experiments}).

The push method has become a frequently-used graph mining primitive because of the sparsity of the vectors that result from when push is used to approximate the seeded PageRank diffusion, along with the speed at which they can be computed.
The method is typically used to identify sets of low-conductance in a graph as part of a community or cluster analysis~\cite{Ghosh-2014-cheeger,Gleich-2012-neighborhoods,Gutierrez-Bunster-2014-biology-networks,jeub2015locally,Leskovec-2009-community-structure,Whang-2013-overlapping}.
In these cases, the insights provided by the solution paths are unlikely to be necessary. Rather, what is needed is a faster way to compute these diffusions for many values of $\eps$.
We describe a data structure called a \emph{shelf} that we demonstrate
can use 40 times as many values of $\eps$ in only 7 times the runtime (Section~\ref{sec:runtime-grid}).

We plan to make our computational codes available in the spirit of reproducible research.%

\section{Technical Preliminaries}\label{sec:prelims}
We first fix our notation and review the Andersen-Chung-Lang procedure, which forms the basis for many of our contributions.  We denote a graph by $G = (V,E)$, where $V$ is the set of nodes and $E$ the set of edges. All graphs we consider are simple, connected, and undirected. Let $G$ have $n = |V|$ nodes and fix a labeling of the graph nodes using the numbers $1$, $2$, $\dots$, $n$. We refer to a node by its label. For each node $j$ we denote its degree by $d_j$.  

The \emph{adjacency matrix} of the graph $G$, which we denote by $\mA$, is the $n\times n$ matrix having $A_{i,j} = 1$ if nodes $i$ and $j$ are connected by an edge, and 0 otherwise. Since $G$ is simple and undirected, $\mA$ is symmetric with 0s on the diagonal. The matrix $\mD$ denotes the diagonal matrix with entry $(i,i)$ equal to the degree of node $i$, $d_i$. Since $G$ is connected, $\mD$ is invertible, and we can define the \emph{random walk transition matrix} $\mP := \mA\mD\inv$. 

We denote by $\ve_j$ the standard basis vector of appropriate dimensions with a 1 in entry $j$, and by $\ve$ the vector of all 1s.
In general, we use subscripts on matrices and vectors to denote entries, e.g. $A_{i,j}$ is entry $(i,j)$ of matrix $\mA$; the notation for standard basis vectors, $\ve_j$, is an exception. Superscripts refer to vectors in a sequence of vectors, e.g. $\vvk{x}{k}$ is the $k$th vector in a sequence.

For any set of nodes, $S \subseteq V$, we define the \emph{volume} of $S$ to be the sum of the degrees of the nodes in $S$, denoted $\vol(S) = \sum_{j \in S} d_j$. Next, define the \emph{boundary} of $S \subseteq V$ to be the set of edges that have one endpoint inside $S$ and the other endpoint outside $S$, denoted $\partial(S)$. Finally, the \emph{conductance} of $S$, denoted $\phi(S)$, is defined by
\[
\phi(S) := \frac{|\partial(S)|}{\min\{\vol(S), \vol(V-S) \} }.
\]
Conductance can be thought of as measuring the extent to which a set is more connected to itself than the rest of the graph and is one of the most commonly used community detection objectives~\cite{Schaeffer-2007-clustering}.

\subsection{PageRank and Andersen-Chung-Lang Method}

The Andersen-Chung-Lang method uses PageRank vectors to identify a set of small conductance focused around a small set of starting nodes~\cite{andersen2006-local}. We call such starting nodes \emph{seed sets} and the resulting communities, \emph{local} communities. We now briefly review this method starting with PageRank.

For a stochastic matrix $\mP$, a stochastic vector $\vv$, and a parameter $\alpha \in (0,1)$ we define the PageRank diffusion as the solution $\vx$ to the linear system
 \begin{equation}\label{eqn:prls}
 (\mI - \alpha \mP) \vx = (1-\alpha)\vv.
 \end{equation}
Note that when $\alpha \in (0,1)$ the system in \eqref{eqn:prls} can be solved via a Neumann series expansion, and so the solution $\vx$ to this linear system is equivalent to the PageRank diffusion vector described in Section~\ref{sec:intro}. 
When $\vv = (1/|S|)\ve_S$, i.e. the indicator vector for a seed set $S$, normalized to be stochastic, then we say the PageRank vector has been \emph{seeded} on the set $S$ (or \emph{personalized} on the set $S$). 

Given PageRank diffusion scores $\vx$, the Andersen-Chung-Lang procedure uses the values $\vx_j/d_j$ to determine an order for a sweep-cut procedure (described below) that identifies a set of good conductance. Thus, we would like to bound the error in approximating the values $\vx_j/d_j$. Specifically (for their theory) we need our approximate solution $\hvx$ to satisfy 
\begin{equation} 0 \leq \vx_j - \hvx_j < \eps d_j \qquad \text{or equivalently,} \qquad 
\label{eqn:conv-vec}
\vx \geq \hvx,  \text{ and } \|\mD\inv(\vx - \hvx)\|_{\infty} < \eps.
\end{equation}

Once a PPR diffusion $\vx$ is computed to this accuracy, a near-optimal conductance set located nearby the seed nodes is generated from the following a \emph{sweep cut} procedure.
Rank the nodes in descending order by their scaled diffusion scores $\vx_j/d_j$ , with large scores ranking the highest.
Denote the set of nodes ranked 1 through $m$ by $S(m)$. Iteratively compute the conductance of the sets $S(m)$ for $m = 2$, $3$, $\dots$, until $\vx_m/d_m =0$. Return the set $S(t)$ with the minimal conductance. This returned set is related to the optimal set of minimum conductance nearby the seed set through a localized Cheeger inequality~\cite{andersen2006-local}. The value of $\eps$ relates to the possible size of the set. 

\section{The push procedure}
\label{sec:prelim-push}
\label{sec:alg-push}
The push procedure is an iterative algorithm to compute a PageRank vector to satisfy the approximation~\eqref{eqn:conv-vec}.  The distinguishing feature is that it can accomplish this goal with a sparse solution vector, which it can usually generate without ever looking at the entire graph or matrix. This procedure allows the Andersen-Chung-Lang procedure to run without ever looking at the entire graph. As we discussed in the introduction, this idea and method are at the heart of our contributions and so we present the method in some depth. 

At each step, push updates only a single coordinate of the approximate solution like a coordinate relaxation method. We'll describe its behavior in terms of a general linear system of equations. Let $\mM\vx = \vb$ be a square linear system with 1s on the diagonal, i.e.~$M_{i,i} = 1$ for all $i$.
Consider an iterative approximation $\vvk{x}{k} \approx \vx$ after $k$ steps. The corresponding residual is $\vvk{r}{k} = \vb - \mM \vvk{x}{k}$. Let $j$ be a row index where we want to \emph{relax}, i.e.~locally solve, the equation, and let $r$ be the residual value there, $r = \vvk{r}{k}_j$.
We update the solution by adding $r$ to the corresponding entry of the solution vector, $\vvk{x}{k+1} = \vvk{x}{k} + r\ve_{j}$, in order to guarantee $\vvk{r}{k+1}_j = 0$. The residual can be efficiently updated in this case. Thus, the push method involves the operations: 
\begin{align}
	\vvk{x}{k+1} &= \vvk{x}{k} + r\ve_{j} \nonumber \\
	\vvk{r}{k+1} &= \vvk{r}{k} - r\mM\ve_{j}. \label{eqn:resupdate}
\end{align}

Note that the iteration requires updating just one entry of $\vvk{x}{k}$ and accessing only a single column of the matrix $\mM$. It is this local update that enables push to solve the seeded PageRank diffusion especially efficiently. 

\subsection{The Andersen-Chung-Lang Push Procedure for PageRank}
\label{sec:acl-push-proc}
The full algorithm for the push method applied to the PageRank linear system to compute a solution that satifies~\eqref{eqn:conv-vec} for a seed set $S$ is:
\begin{enumerate}
 \item[ 1. ]  Initialize $\vx = 0, \vr = (1-\alpha) \ve_S$ using sparse data structures such as a hash-table.
 \item[ 2. ] Add any coordinate $i$ of $\vr$ where $\vr_i \ge \eps d_i$ to a queue $Q$.
 \item[ 3. ] While $Q$ is not empty
 \item[ 4. ] \hspace*{1em} Let $j$ be the coordinate at the front of the queue and pop this element.
 \item[ 5. ] \hspace*{1em} Set $\vx_j = \vx_j + \vr_j$
 \item[ 6. ] \hspace*{1em} Set $\delta = \alpha \vr_j / d_j$
 \item[ 7. ] \hspace*{1em} Set $\vr_j = 0$
 \item[ 8. ] \hspace*{1em} For all neighbors $u$ of node $j$
 \item[ 9. ] \hspace*{2em} Set $\vr_u = \vr_u + \delta$
 \item[10. ] \hspace*{2em} If $\vr_u$ exceeds $\eps d_u$ after this change, add $u$ to $Q$.
\end{enumerate}
The queue maintains a list of all coordinates (or nodes) where the residual is larger than $\eps d_j$. We choose coordinates to relax from this queue. Then we execute the push procedure to update the solution and residual. The residual update operates on only the nodes that neighbor the updated coordinate $j$. Once elements in the residual exceed the threshold, they are entered into the queue. We present the convergence theory for this method in the description of our new algorithms (Section~\ref{sec:alg}).

We have presented the push method so far from a linear solver perspective.
To instead view the method from a graph diffusion perspective, think of the solution vector as tracking where ``dye'' has concentrated in the graph and the residual as tracking where ``dye'' is still spreading.
At each step of the method, we find a node with a sufficiently large amount of dye left (Step 4), concentrate it at that node (Step 5), then update the amount of dye that is left in the system as a result of concentrating this quantity of dye (Lines 6-10). The name \emph{push} comes from the pattern of concentrating dye and \emph{pushing} newly unprocessed dye to the adjacent residual entries. 

Note that the value of $\eps$ plays a critical role in this method as it determines the entries that enter the queue. When $\eps$ is large, only a small number of coordinates or nodes will ever enter the queue. This will result in a sparse solution. As $\eps \to 0$, there will be substantially more entries that enter the queue.


\subsection{Implicit regularization from Push}
\label{sec:regularization}
To understand the sparsity that results from the push method, we introduce a slight variation on the standard push procedure.
Rather than using the full update $\vx_j + \vr_j$ and pushing $\alpha \vr_j / d_j$ to the adjacent residuals, we consider a method that takes a partial update. The form we assume is that we will leave $\eps d_j \rho$ ``dye'' remaining at node $j$. For $\rho = 0$, this correspond to the push procedure described above. For $\rho = 1$, this update will remove node $j$ from the queue, but push as little mass as possible to the adjacent nodes such that the dye at node $j$ will remain below $\eps d_j$.
The change is just at steps 5-7:
\begin{enumerate}
 \item[ 5'. ] \hspace*{1em} Set $\vx_j = \vx_j + (\vr_j - \eps d_j \rho)$
 \item[ 6'. ] \hspace*{1em} Set $\delta = \alpha (\vr_j - \eps d_j \rho) / d_j$
 \item[ 7'. ] \hspace*{1em} Set $\vr_j = \eps d_j \rho$
\end{enumerate}
In previous work~\cite[Theorem 3]{Gleich-2014-alg-anti-diff}, we showed that $\rho = 1$ produces a solution vector $\vx$ that exactly solves a related 1-norm regularized optimization problem. The form of the problem that $\vx$ solves is most cleanly stated as a quadratic optimization problem in $\vz$, a degree-based rescaling of the solution variable $\vx$:
\begin{equation} \label{eq:ppr-regularized}
 \begin{array}{ll}
  \text{minimize} & \displaystyle \frac{1}{2} \vz^T \mQ \vz - \vz^T \vg + C \eps \|{\mD \vz}\|_1 \\
  \text{subject to} & \vz \ge 0
 \end{array}
\end{equation}
The terms of the normalization $\vx$ vs.~$\vz$ and the equivalence $\mQ, \vg, C$ are tedious to state exactly and uninformative to our purposes in this work. The important point is that $\eps$ can also be interpreted as a regularization parameter that governs the sparsity of the solution vector $\vx$. Large values of $\eps$ increase the magnitude of the 1-norm regularizer and thus cause the solutions to be sparser. Moreover, the resulting solutions are unique as the above problem is strongly convex. 

In this work, we seek algorithms to compute the solution paths or regularization paths that result from trying to use all values of $\eps$ to fully study the behavior of the diffusion. In the next section we explore some potential utilities of these paths before presenting our algorithms for computing them in Section~\ref{sec:alg}.

\section{Personalized PageRank paths}
\label{sec:paths}
\label{sec:path-analysis}

In this section we aim to show the types of insights that our solution path methodology can provide.
We should remark that these are primarily designed for human interpretation. Our vision is that they would be used by an analyst that was studying a network and needed to better understand the ``region'' around a target node. These solution paths would then be combined with something like a graph layout framework to study these patterns in the graph. Thus, much of the analysis here will be qualitative. We demonstrate quantative advantages to the path methodology in subsequent sections. 

\subsection{Exact paths and fast path approximations}
The exact solution path for the seeded PageRank diffusion results from solving the regularized optimization problem~\eqref{eq:ppr-regularized} itself for all values of $\eps$. This could be accomplished by using ideas similar to those used to compute solution paths for the Lasso regularizer~\cite{Efron-2004-lars}. Our algorithms and subsequent analysis evaluate approximate solution paths that result from using our push-based algorithm with $\rho=0.9$ (Section~\ref{sec:alg-path}). In this section, we compare these approximate paths to the exact paths. We find that, while the precise numbers change, the qualitative properties are no different. 

Figure~\ref{fig:netsci-exact} shows the results of such a comparison on Newman's netscience dataset (379 nodes, 914 edges~\cite{Newman-2006-eigenvectors}).
Each curve or line in the plot represents the value of a non-zero entry of an approximate PageRank vector $\vx_{\eps}$ as $\eps$ varies (horizontal axis). As $\eps$ approaches 0 (and $1/\eps$ approaches $\infty$), each approximate PageRank entry approaches its exact value in a monotonic manner. Alternatively, we can think of each line as the diffusion value of a node as the diffusion process spreads across the graph. 

One of the plots was computed by solving for the optimality conditions of~\eqref{eq:ppr-regularized}; the other plot was computed using the PPR path algorithm from Section~\ref{sec:alg-path}. The values of $\eps$ are automatically determined by the algorithm itself. The plots show that for the two sets of paths have essentially identical qualitative features. For example, they reveal the same bends and inflections in individual node trajectories, as well as large gaps in PageRank values. The maximum difference between the two paths never exceeds $1.1\cdot10^{-4}$.

These results were essentially unchanged for a variety of other sample diffusions we considered, and so we decided that using $\rho=0.9$ was an acceptable compromise between speed and exactness. Thus, all path plots in this paper were created with $\rho=0.9$, unless noted otherwise. (For analysis of the \emph{differences} of the exact paths and $\rho$-paths, and in particular the behavior of the $\rho$-approximate paths as $\rho$ varies, see Figure~\ref{fig:rho-scaling} below.)

\begin{figure*}[h]
\centering
 \includegraphics[width=0.5\linewidth]{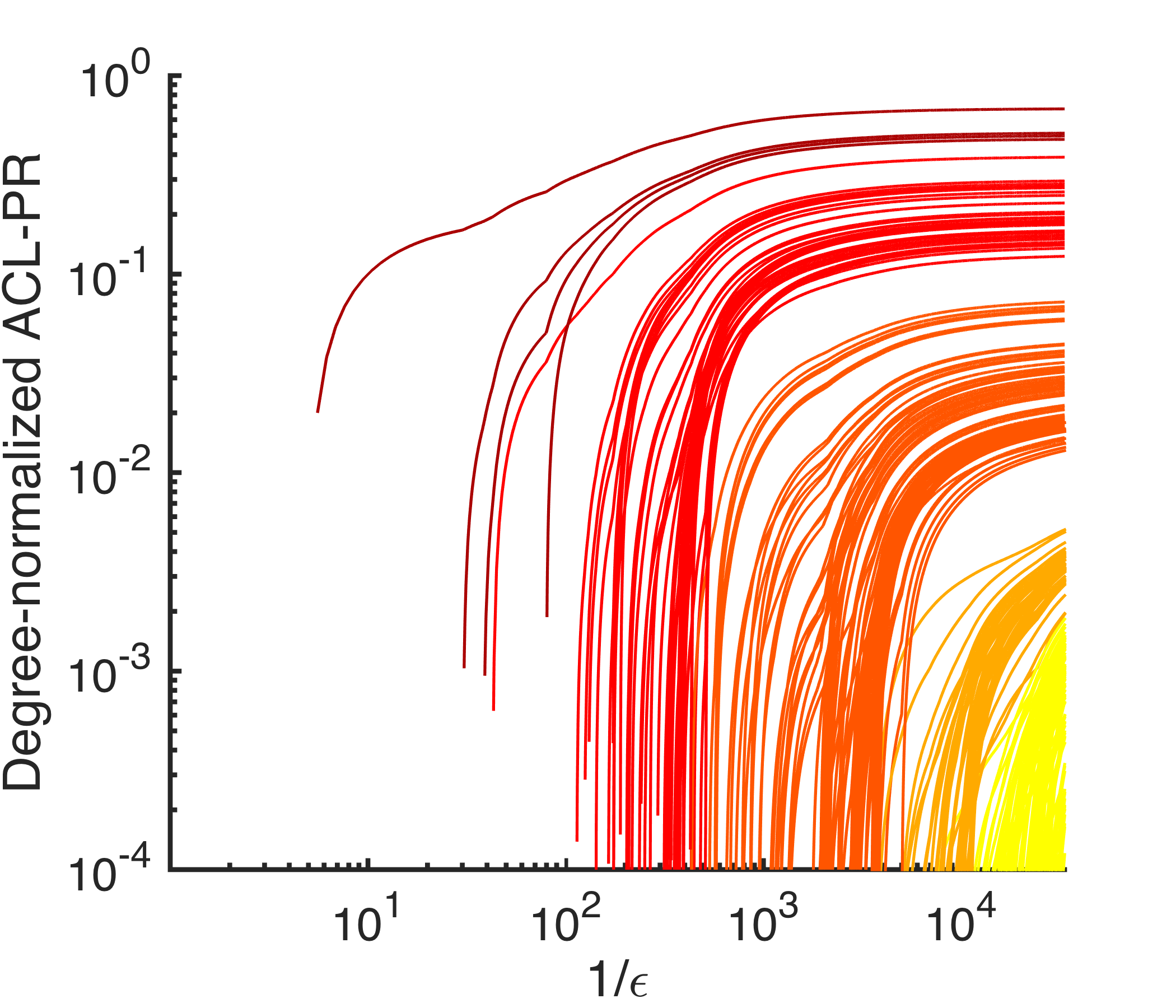}%
 \includegraphics[width=0.5\linewidth]{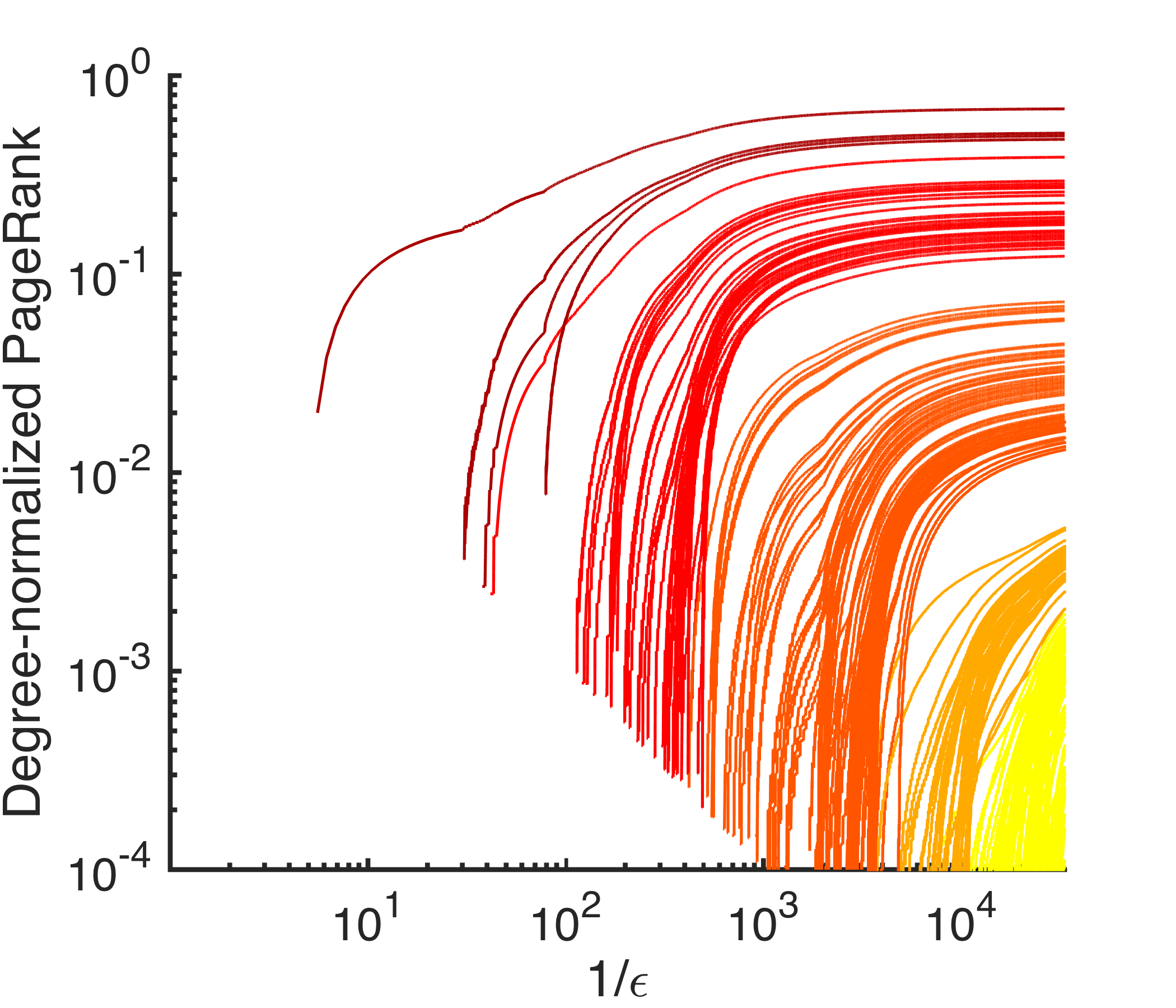}%
\caption{(Left) The solution paths for a PageRank diffusion on Newman's netscience dataset from a single seed node computed by exactly solving the regularized problem. (Right) The approximate solution paths computed by our push-based solution path algorithm with $\rho=0.9$. Each line traces a value $\vx_j$ as $\eps$ varies. The maximum infinity-norm distance between the two paths is $1.1\cdot10^{-4}$, showing that $\rho=0.9$ provides a good qualitative approximation.  Moreover, the two plots highlight identical qualitative features---for example, the large gaps between paths, and the strange bend in the paths near $\eps = 10^{-3}$. 
The coloring of the lines is based on the values at the smallest value of $\eps$. The values of $\eps$ used were generated by the approximate algorithm itself and we computed the exact solution at these same values for comparision.
}
\label{fig:netsci-exact}
\end{figure*}

\subsection{The Seeded PageRank Solution Path Plot}
We now wish to introduce a specific variation on the solution path plot that shows helpful contextual information. In the course of computation, our solution path algorithm identifies a small set of values of $\eps$ (somewhere between a few hundred to a few thousand) where it satisfies the solution criteria~\eqref{eqn:conv-crit}. At these values, we perform a sweep-cut procedure to identify the set of best conductance induced by the current solution. In the solution path plot, we display the cut-point identified by this procedure as a thick black line. All the nodes whose trajectories are above the dark black line at a particular value of $\eps$ are contained in the set of best conductance at that value of $\eps$. This line allows us to follow the trajectory of the minimum conductance set as we vary $\eps$.  Another property of our algorithm is that the smallest possible non-zero diffusion value in the solution is $(1-\rho)\eps$. Thus, we plot this as a thin, diagonal, black line that acts as a pseudo-origin for all of the node trajectories.  The vertical blue lines in the bottom left of the plot mark the values of $\eps$ where we detect a significant new set of best conductance. Representative conductance values are shown when there is room in the plot.

The solution path plot that corresponds to Figure~\ref{fig:netsci-exact} is shown in Figure~\ref{fig:solution-plot}. This plot illustrates all of the features we discussed in this section.

\begin{figure}[t]
 \includegraphics[width=0.65\linewidth]{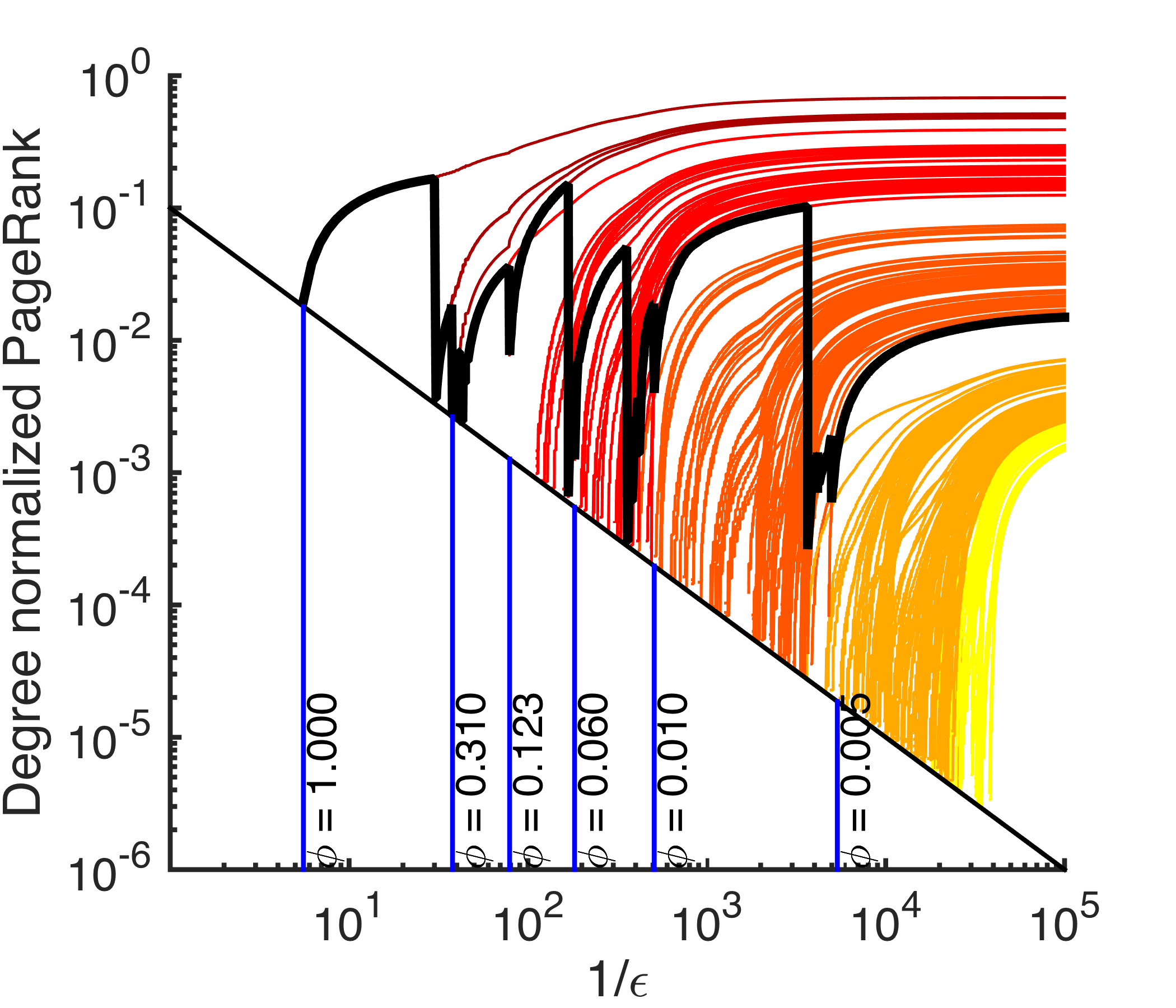}
 \caption{An example of the seeded PageRank solution path plot on Newman's netscience dataset. Each colored line represents the value of a single node as the diffusion progresses from large $\eps$ to small $\eps$. Because of our $\rho$-approximation to the true paths, the smallest value any node obtains is $(1-\rho) \eps$ and we plot this as a dark diagonal line. The thick black line traces out the boundary of the set of best conductance found at each distinct value of $\eps$ as determined by a sweep-cut procedure. The blue lines indicate significant changes to the set of minimum conductance, and they are labelled with the conductance value. The coloring of the trajectory lines is based on the values at the smallest value of $\eps$.  We discuss implications of the plot in Section~\ref{sec:netsci-youtube-paths}.}
 \label{fig:solution-plot}
 \label{fig:netsci-path}
\end{figure}

%
\subsection{Nested communities in netscience and Facebook}
\label{sec:netsci-youtube-paths}

We now discuss some of the insights that arise from the solution path plot. 
In Figure~\ref{fig:netsci-path}, we show the seeded PageRank solution path plot for around $21,000$ values of $\eps$ computed via our algorithm for the network science collaboration network. This computation runs in less than a second. Here, we see that large gaps in the degree normalized PageRank vector indicate cutoffs for sets of good conductance. This behavior is known to occur when sets of really good conductance emerge~\cite{Andersen-2007-sharp-drops}. We can now see how they evolve and how the procedure quickly jumps between them. In particular,  the path plots reveal multiple communities (good conductance sets) nested within one another through the gaps between the trajectories. 

On a crawl of a Facebook network from 2009 where edges between nodes correspond to observed interactions~\cite{Wilson-2009-social-networks} (see Table~\ref{tab:datasets}, \texttt{fb-one}, for the statistics), we are able to find a large, low conductance set using our solution path method. (Again, this takes about a second of computation.) Pictured in Figure~\ref{fig:fbA-path}, this diffusion shows no sharp drops in the PageRank values like in the network science data, yet we still find good conductance cuts. Note the few stray ``orange'' nodes in the sea of yellow. These nodes quickly grow in PageRank and break into the set of smallest conductance. Finding these nodes is likely to be important to understand the boundaries of communities in social networks; these trajectories could also indicate anomalous nodes. Furthermore, this example also shows evidence of multiple nested communities. These are illustrated with the manual annotations $A, B, C$. 

\begin{figure}[t]
\centering
\includegraphics[width=0.65\linewidth]{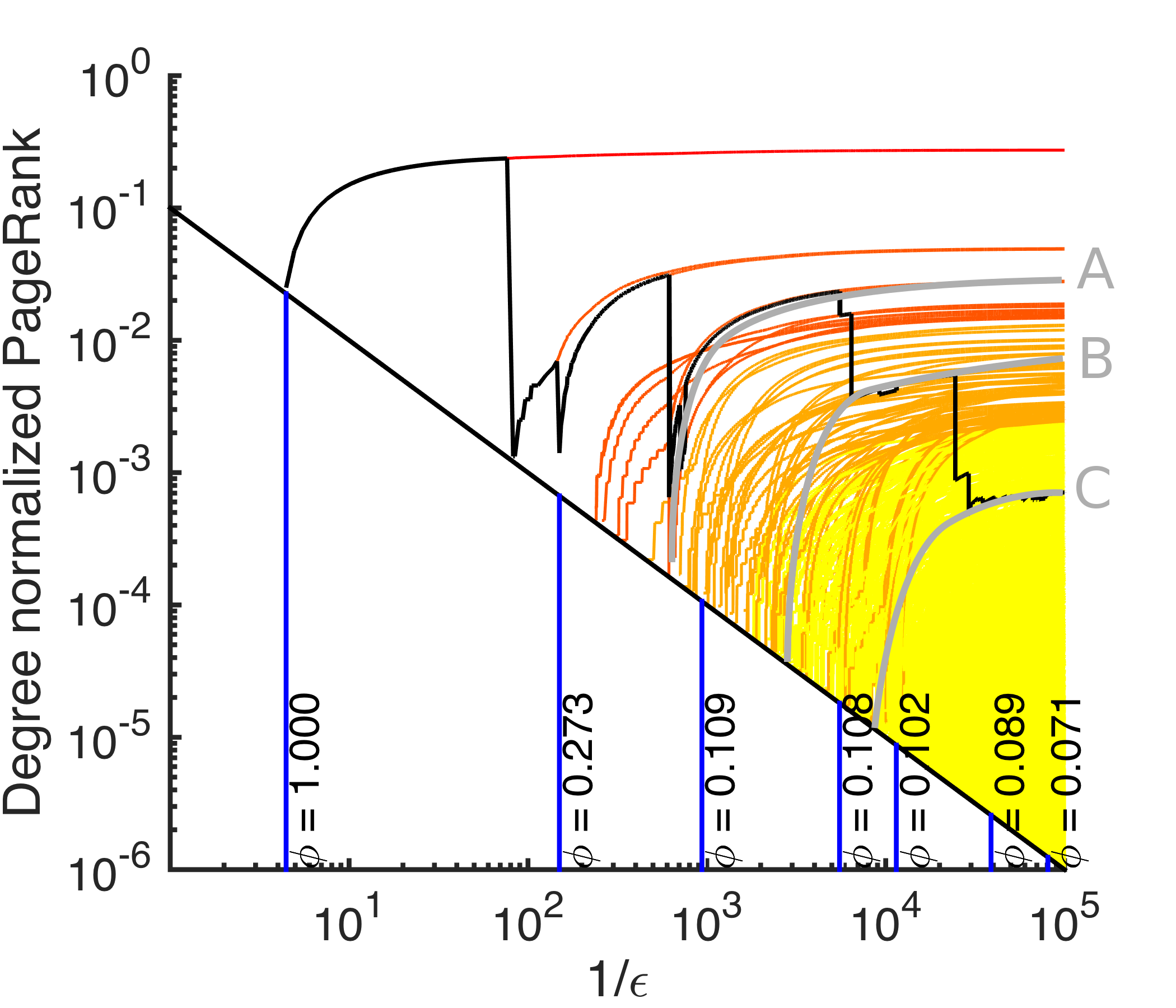}
 \caption{The seeded PageRank solution path for a crawl of observed Facebook network activity for one year (fb-one from Table~\ref{tab:datasets}) shows large, good cuts do not need to have large drops in the PageRank values. Nodes enter the solution and then quickly break into the best conductance set, showing that the frontier of the diffusion should be an interesting set in this graph. Furthermore, this path plot shows evidence of multiple nested communities ($A$, $B$, and $C$), which were manually annotated. The set $A$ is only a few nodes, but has a small conductance score of $0.11$; set $B$ grows and improves this to a conductance of $0.1$, and finally set $C$ achieves a conductance of $0.07$, which is an unusually small conductance value in a large social network.}
 \label{fig:fbA-path}
\end{figure}

\subsection{Core and periphery structure in the US Senate}
\label{sec:senate}

The authors in \cite{jeub2015locally} analyzed voting patterns across the first 110 US-Senates by comparing senators in particular terms. We form a graph from this US Senate data where each senator is represented by a single node. For each term of the senate, we connect senators in that session to their 3 nearest neighbors measured by voting similarities.  This graph has a substantial temporal structure as a senator from 100 years ago cannot have any direct links to a senator serving 10 years ago. We show how our solution paths display markedly different characteristics when seeded on a node near the core of the network compared with a node near the periphery. This example is especially interesting because both diffusions lead to on closely related cuts. 

Figure~\ref{fig:senate} displays solution paths seeded on a senator on the periphery of the network (top right) and a senator connected to the core of the network (top left). Here are some qualitative insights from the solution path plots. 
The peripheral seed is a senator who served a single term; the diffusion spreads across the graph slowly because the seed is poorly connected to the network outside the seed senator's own senate term.
As the diffusion spreads outside the seed's particular term, the paths identify 
multiple nested communities that essentially reflect previous and successive terms of the Senate. 
In contrast, the core node is a senator who served eight terms. The core node's paths skip over such smaller-scale community structures (i.e. individual senate terms) as the diffusion spreads to each of those terms nearly simultaneously. Instead, the paths of the core node identify only one good cut: the cut separating all of the seed's terms from the remainder of the network.

\begin{figure*}[htp!]
 \subfigure[Core seed]{\includegraphics[width=0.5\linewidth]{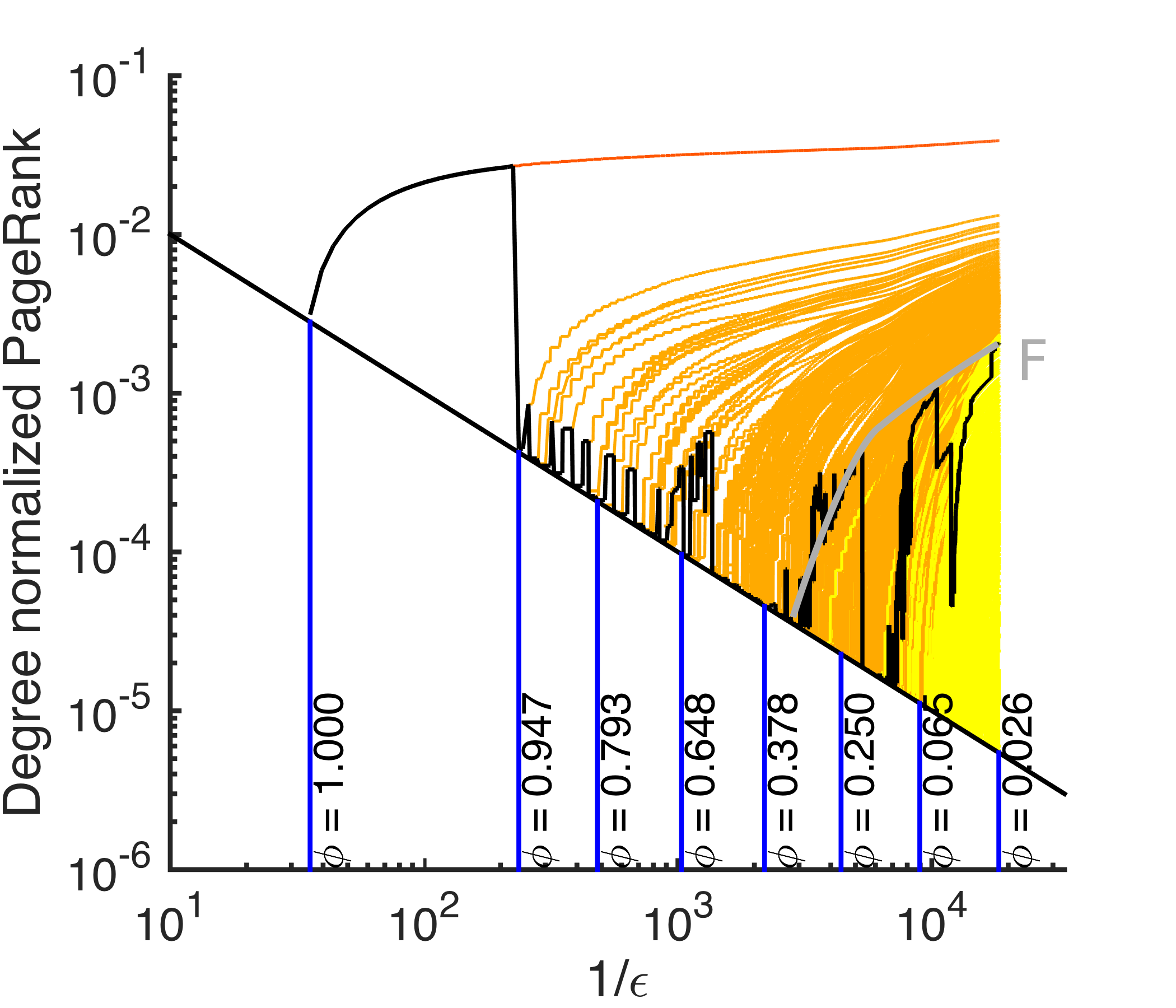}}%
 \subfigure[Periphery seed]{\includegraphics[width=0.5\linewidth]{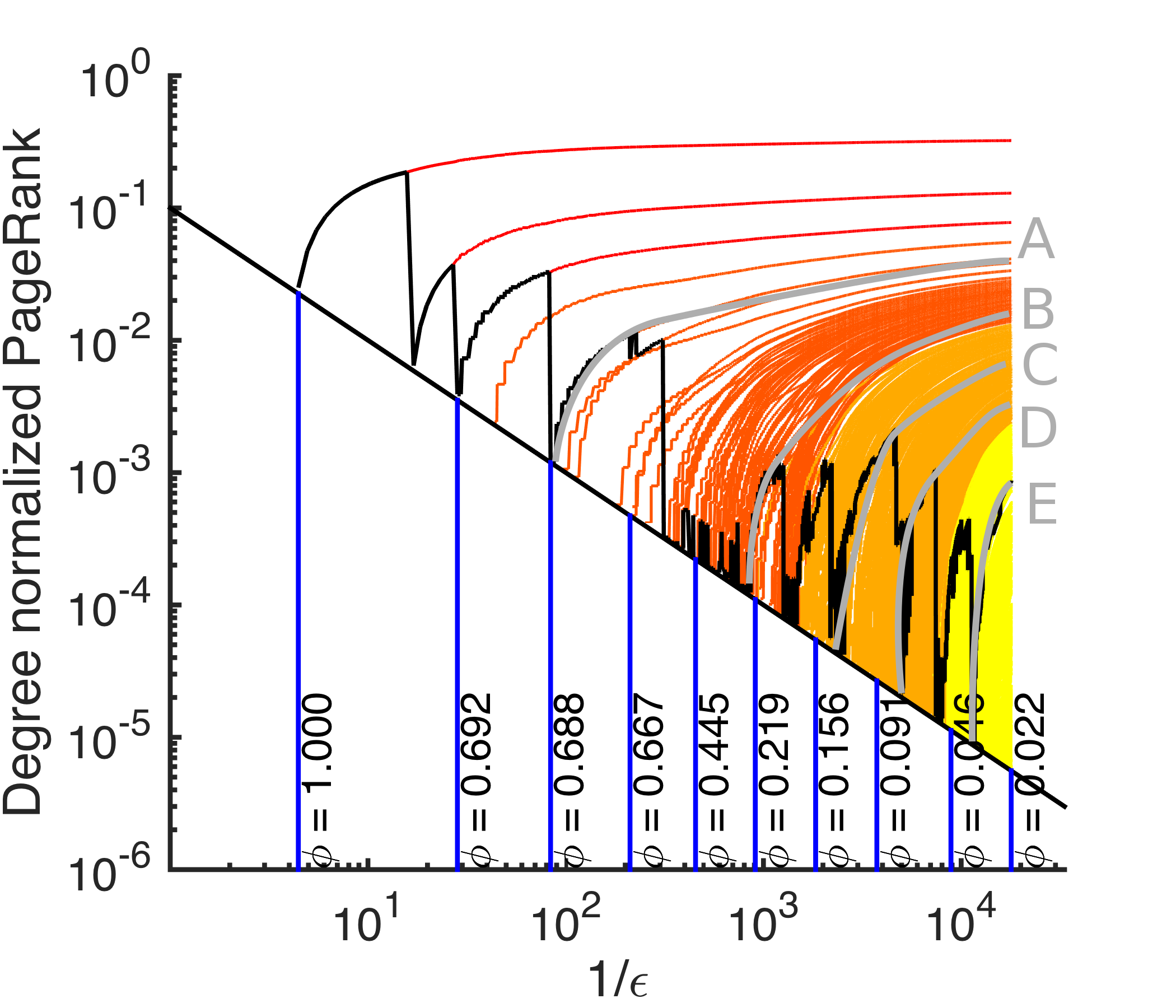}}%
\\
\subfigure[Core, $\eps=3\cdot10^{-4}$]{\includegraphics[width=0.3\linewidth]{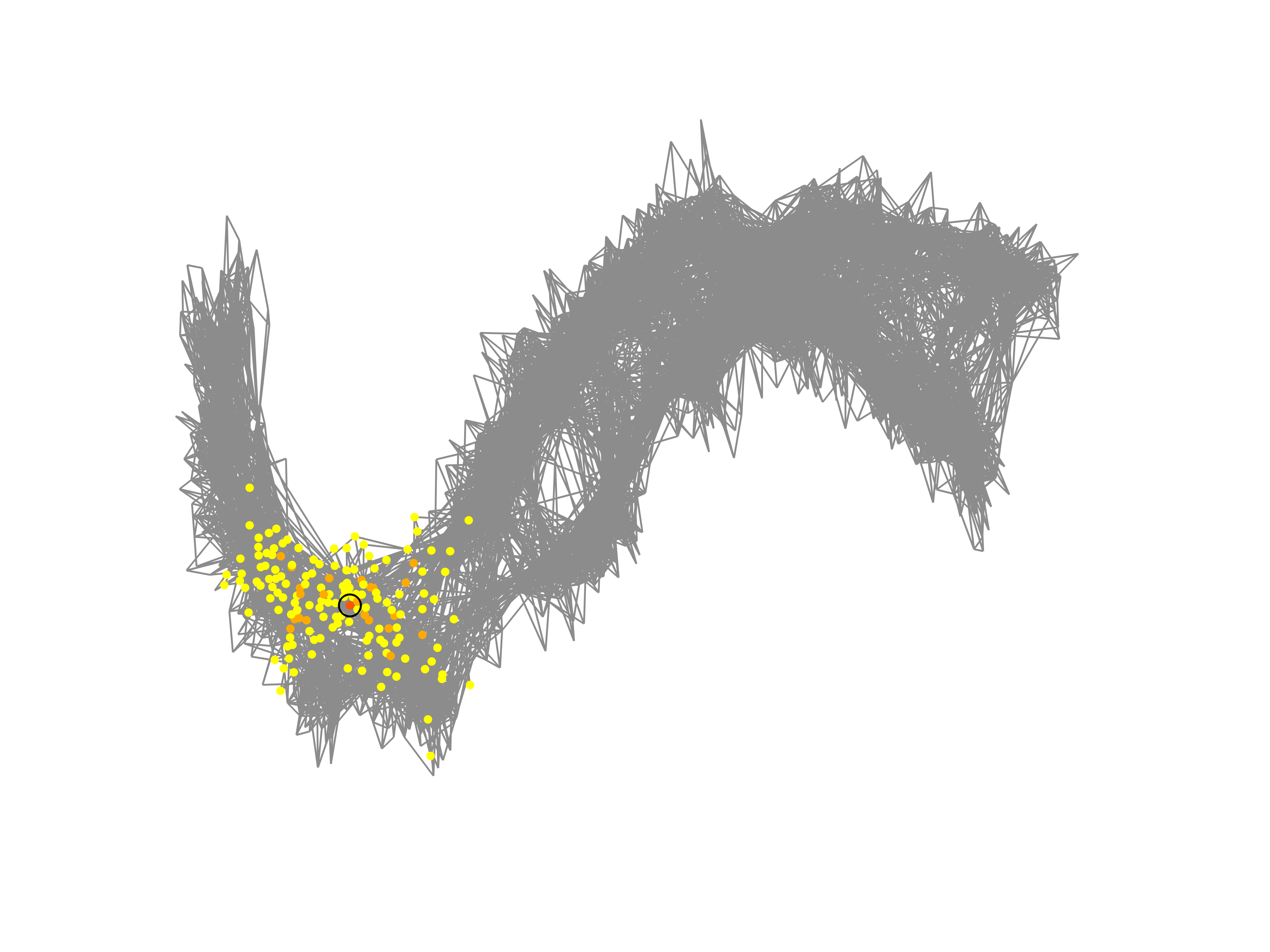}}%
\subfigure[Core, $\eps=10^{-4}$]{\includegraphics[width=0.3\linewidth]{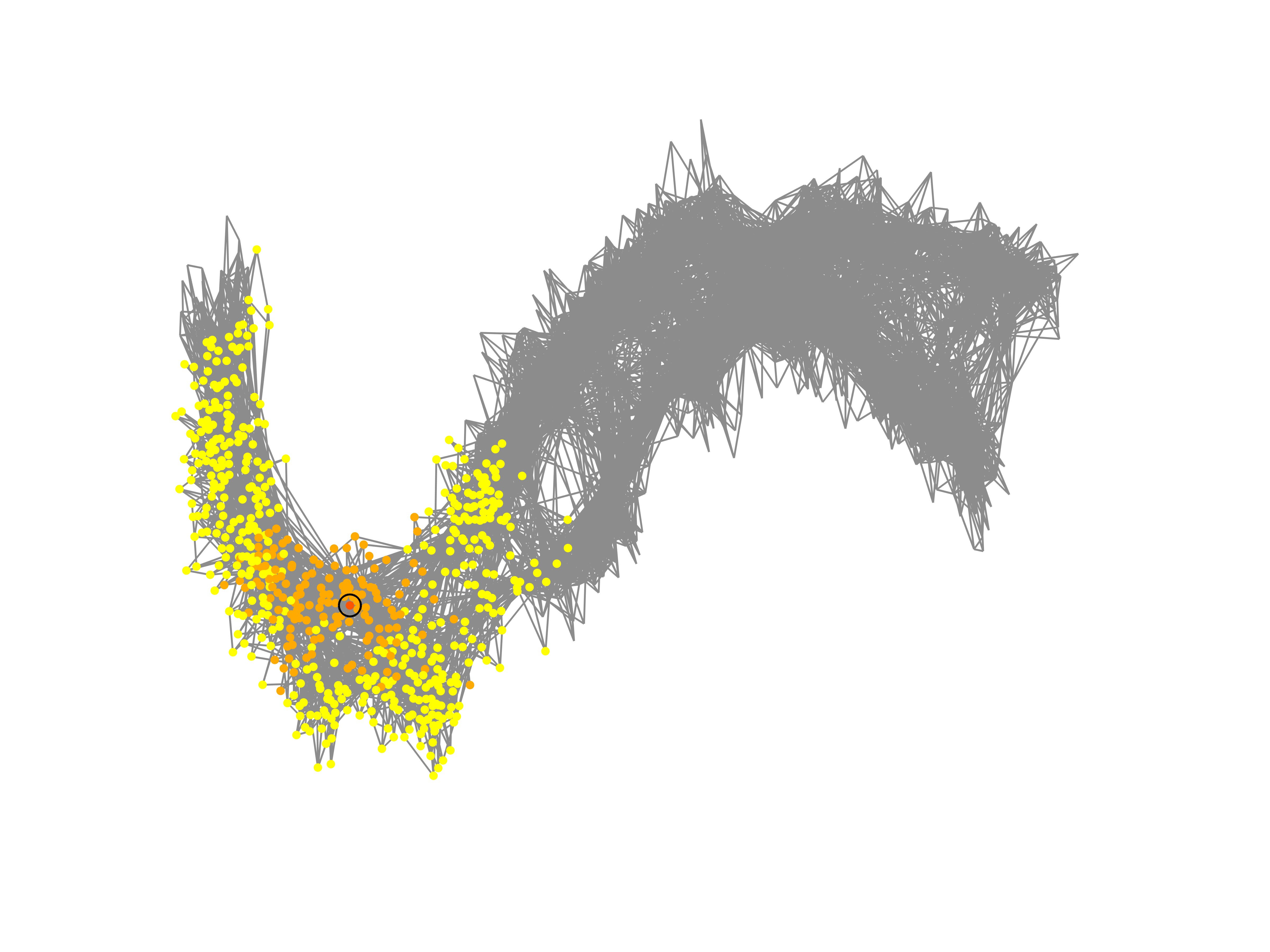}}%
\subfigure[Core, $\eps=3\cdot10^{-5}$]{\includegraphics[width=0.3\linewidth]{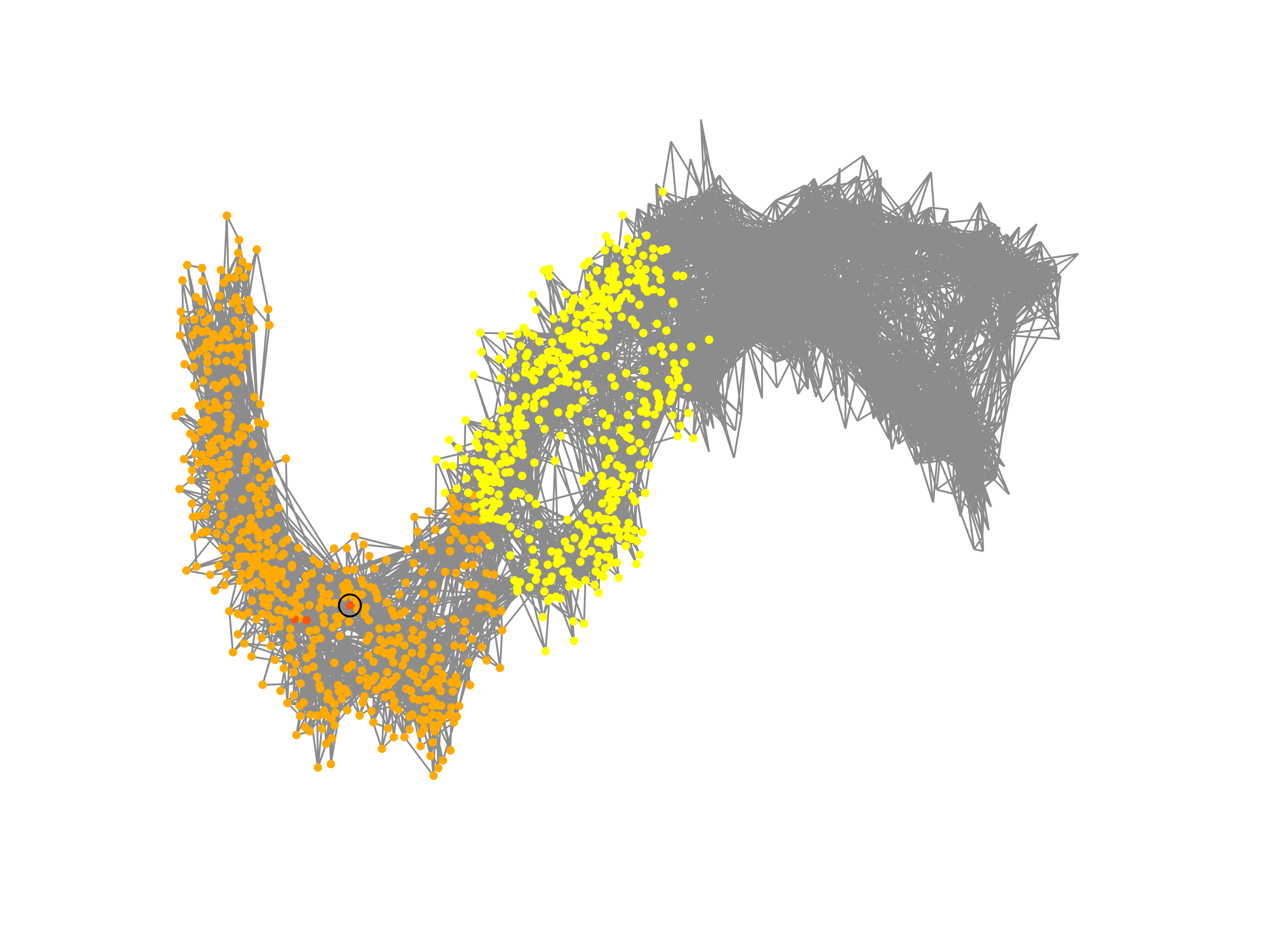}}%
\\
\subfigure[Periphery, $\eps=3\cdot10^{-4}$]{\includegraphics[width=0.3\linewidth]{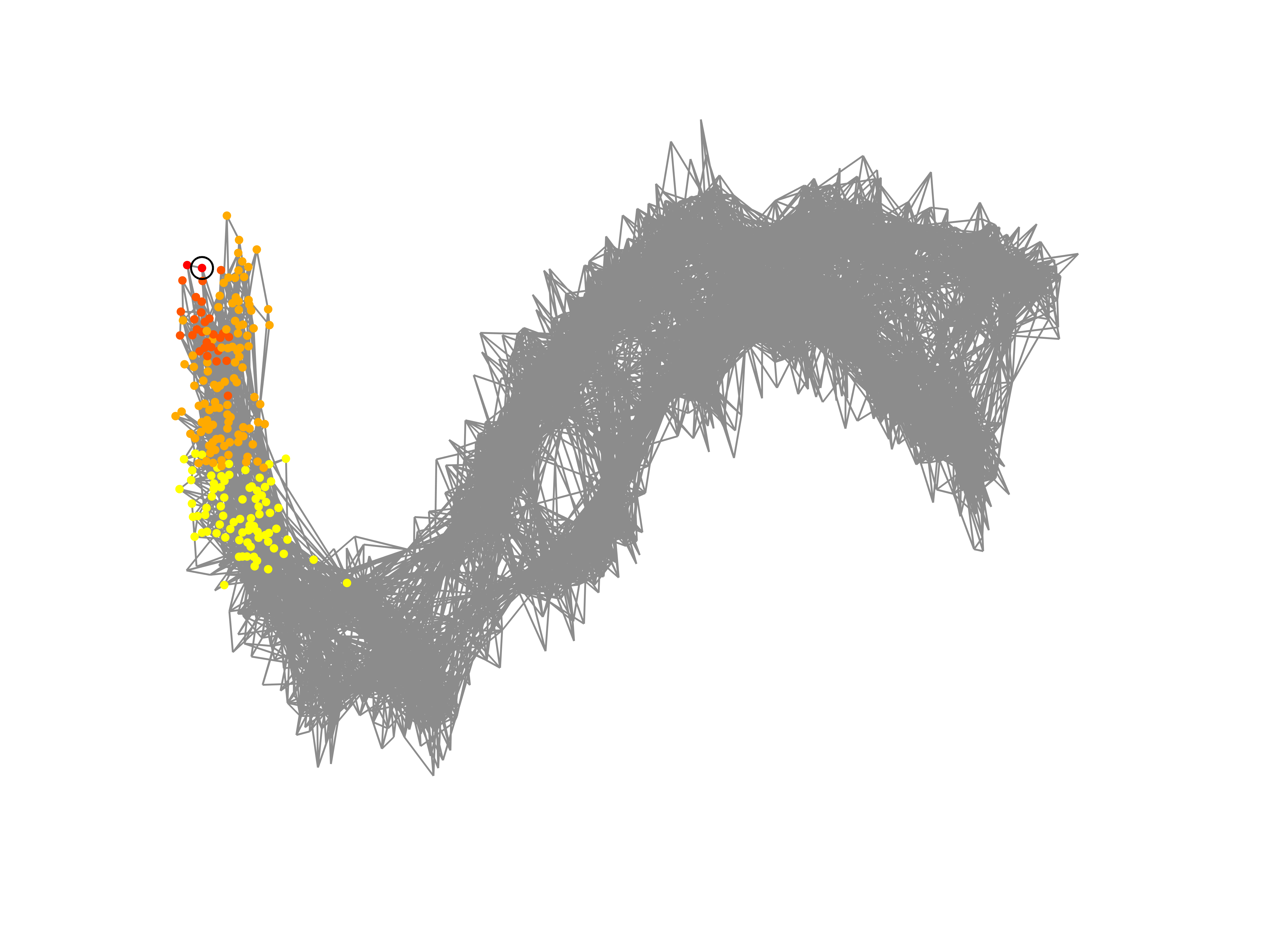}}%
\subfigure[Periphery, $\eps=10^{-4}$]{\includegraphics[width=0.3\linewidth]{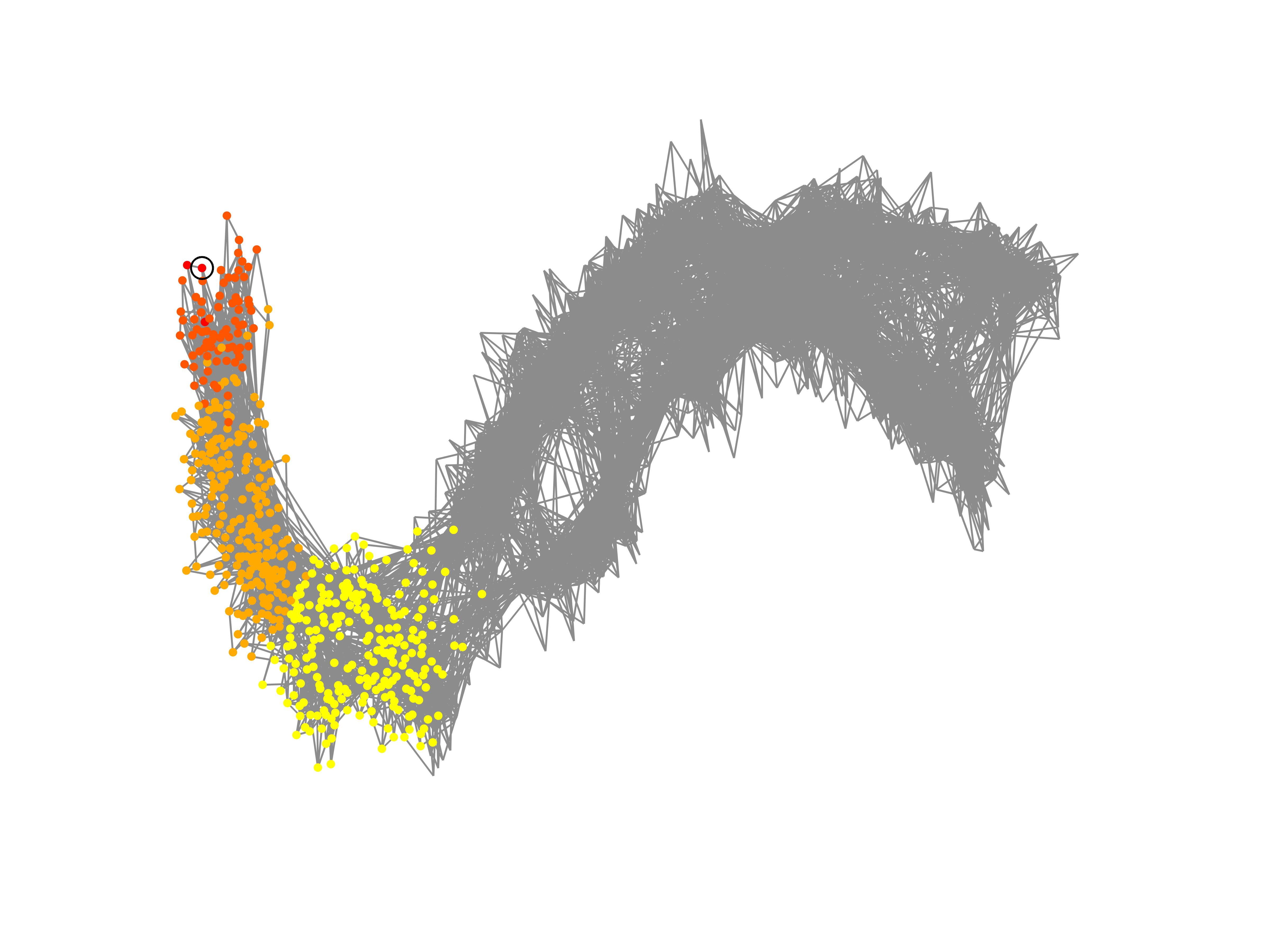}}%
\subfigure[Periphery, $\eps=3\cdot10^{-5}$]{\includegraphics[width=0.3\linewidth]{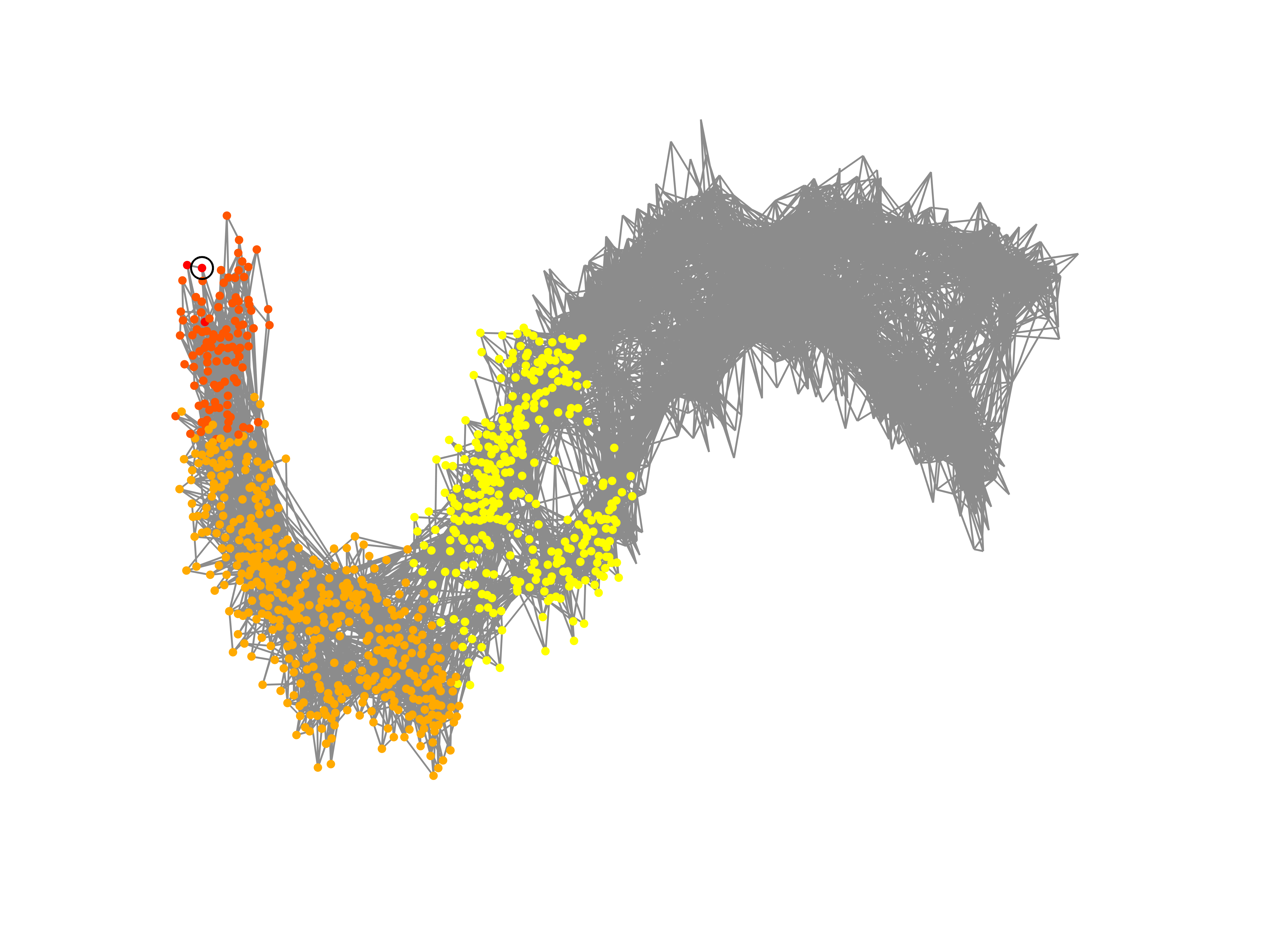}}%
\caption{\emph{(Top.)} The solution paths on the US-Senate graph for a senator in the core (who served multiple terms and is centrally located in a graph layout) and for a senator in the periphery (who served a single term and is located on the boundary of the graph layout). \emph{(Bottom.)} The diffusions for each of these senators are shown as heat-plots on the graph layout. Red indicates nodes with the largest values and yellow the smallest.
The seed nodes are circled in these layouts. 
 The solution paths for a peripheral node indicate multiple nested communities, visible in the images of the diffusion on the whole graph and marked $A$, $B$, $C, D, E$. These sets are strongly correlated with successive terms of the Senate. In contrast, the core node diffusion only indicates one good cut.  For the core node, we can see the
 diffusion essentially spreads across multiple dense regions simultaneously, without settling in one easily separated region until $\eps$ is small enough that the diffusion has spread to the entire left side of the graph. The sets $A$ and $F$ are also almost the same.
}
 \vspace*{-\baselineskip}
 \label{fig:senate-diff}
 \label{fig:senate}
\end{figure*}

This example demonstrates the paths' potential ability to shed light on a seed's relationship to the network's core and periphery, as well as the seed's relationship to many communities.

\subsection{Cluster boundaries in handwritten digit graphs}
\label{sec:usps}

Finally, we use the solution paths to study the behavior of a diffusion for a semi-supervised learning task.
The USPS hand-written digits dataset consists of roughly 10,000 images of the digits $0$ through $9$ in human hand-writing~\cite{zhou2003learning}. Each digit appears in roughly 1,000 of the images, and each image is labelled accordingly. From this data we construct a 3-nearest-neighbors graph, and carry out our analysis as follows. Pick one digit, and select 4 seed nodes uniformly at random from the set of nodes labelled with this digit. Then compute the PageRank solution paths from these seeds. Figure~\ref{fig:usps-path-plot} shows the path plots with labels (right) and without (left). In the labelled plot, the correct labels are red and the incorrect labels are green. 

We can use the best conductance set determined by the PPR vector to capture a number of other nodes sharing the seeds' label.
However, this straight-forward usage of a PageRank vector results in a number of false positives. Figure~\ref{fig:usps-path-plot} (right) shows that a number of nodes with incorrect labels are included in the set of best conductance (curves that are not colored red do not share the seed's label).
 
Looking at the solution-paths for this PageRank vector (Figure~\ref{fig:usps-path-plot}, left) we can see that a number of these false positives can be identified as the erratic lighter-orange paths cutting across the red paths. Furthermore, the solution paths display earlier sets of best conductance (left of the black spikes near $\eps = 10^{-3}$) that would cut out almost all false positives. This demonstrates that the solution paths can be used to identify ``stable'' sets of best conductance that are likely to yield higher precision labeling results. Consequently, these results hint that a smaller, but more precise, set lurks inside of the set of best conductance. This information would be valuable when determining additional labels or trying to study new data that is not as well characterized as the USPS digits dataset. 

\begin{figure*}[th]
\centering
\includegraphics[width=0.5\linewidth]{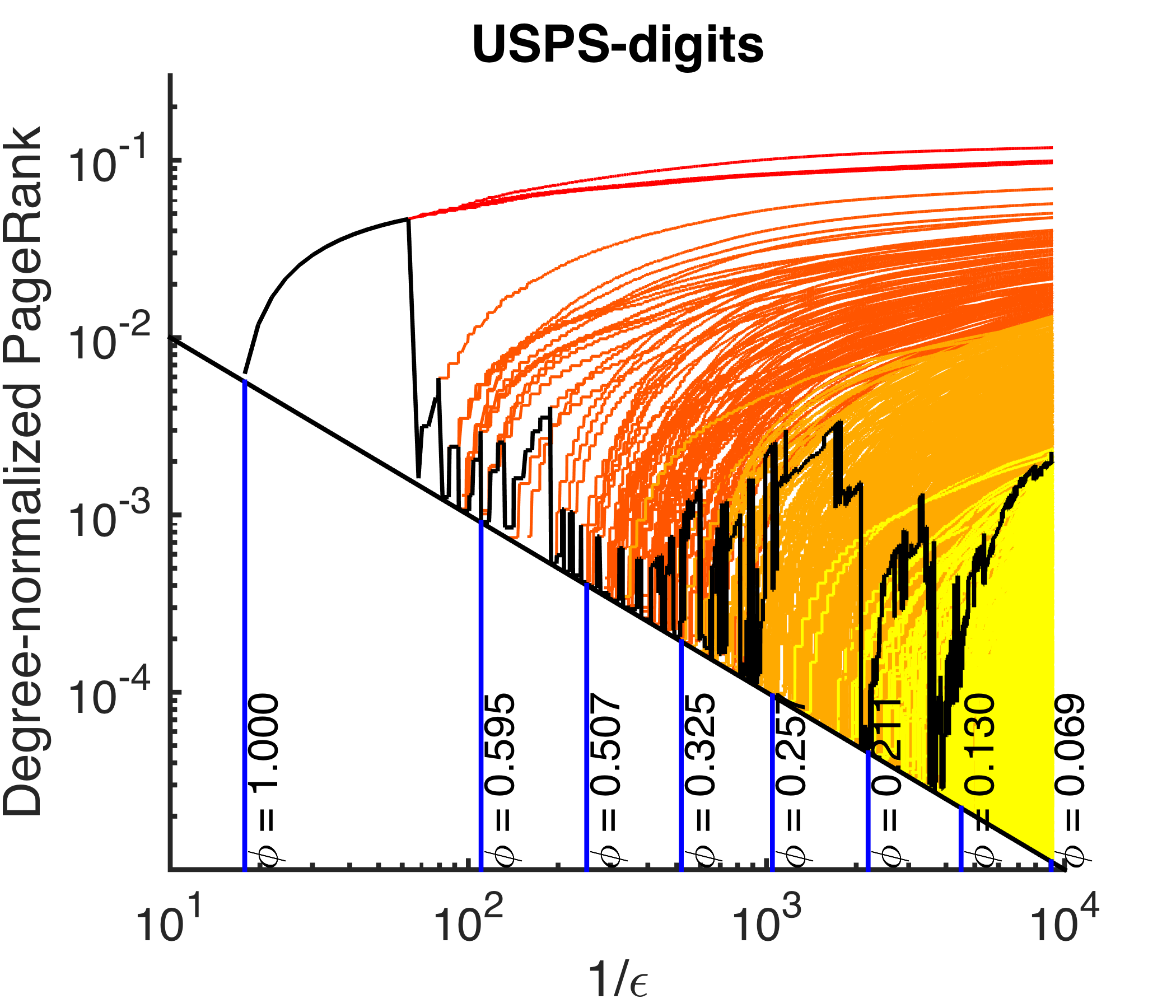}%
\includegraphics[width=0.5\linewidth]{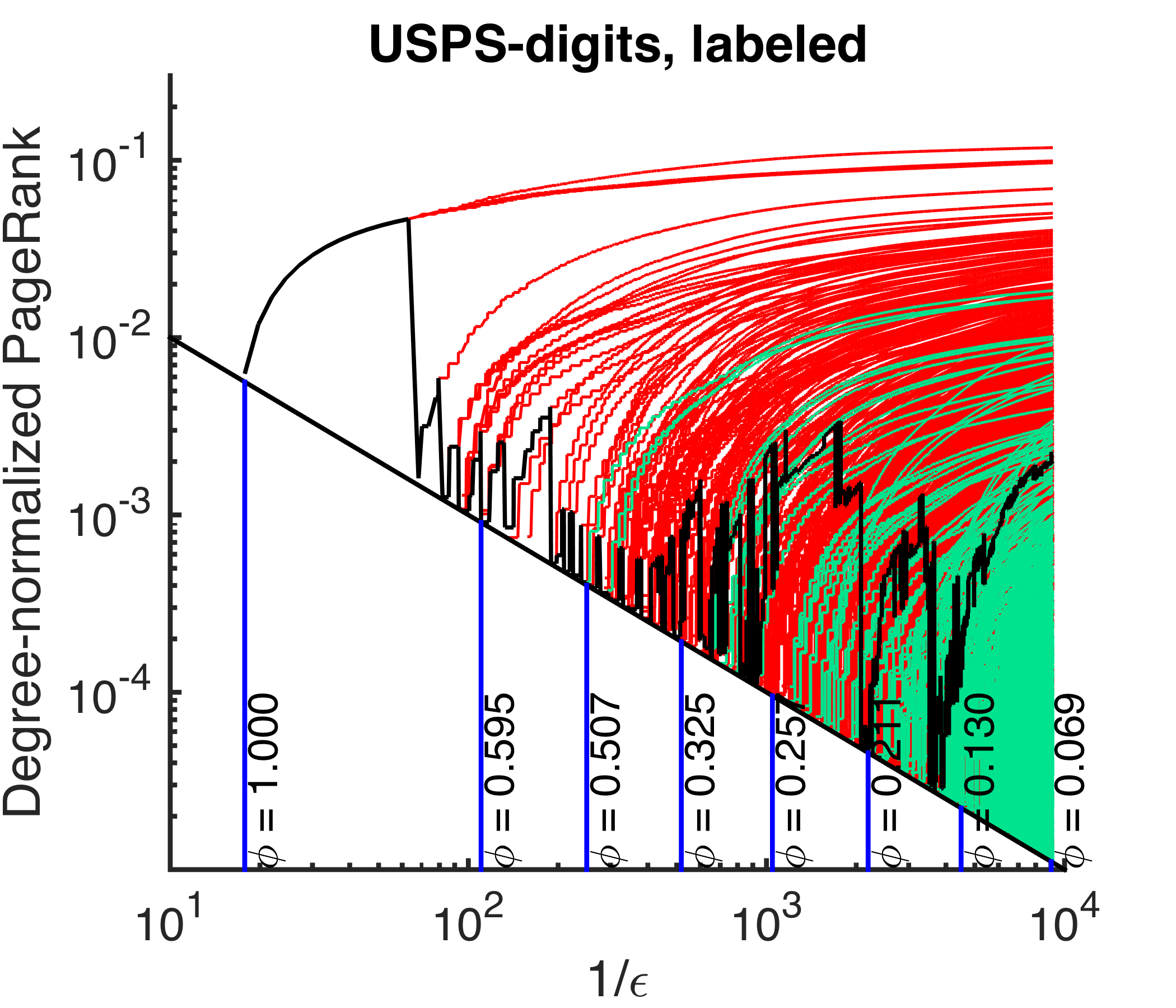}%
 \caption{Seeded PageRank solution path plots for diffusions in the USPS digit dataset. The seeds are chosen to be images of handwritten digits with the same label. \emph{(At left.)} The solution paths reveal a number of anomalous node trajectories near the set of best conductance. Nodes entering the set of best conductance after the black line erratically oscillates are most likely to be false positives near the boundary. 
 \emph{(At right.)} Here, we have colored the solution path lines based on the true-class label. Red shows a correct label and green shows an incorrect label. 
 }
 \label{fig:usps-path-plot} 
\end{figure*}

\subsection{Discussion}

Overall, these seeded PageRank solution path plots reveal information about the clusters and sets near the seeds. Some of the features we've seen include nested community structure and core-periphery structure. They all provide refined information about the boundary of a community containing the seed, and suggest nodes with seemingly anomalous connections to the seed.  For instance, some nodes enter the diffusion early but have only a slow-growing value indicating a weak connection to the seed; other nodes are delayed in entering the diffusion but quickly grow in magnitude and end up being significant members of the cluster. 
Each of these features offers refined insights over the standard single-shot diffusion computation.

\section{Algorithms}\label{sec:alg}

Here we present two novel algorithms for analyzing a PPR diffusion across a variety of accuracy parameter settings by computing the diffusion only a single time. Our first algorithm (Section~\ref{sec:path}) computes the best-conductance set from the $\rho$-approximate solution paths described in Section~\ref{sec:regularization}.
This effectively finds the best-conductance set from
PPR diffusions for \emph{every} accuracy satisfied in an interval $[\epsmn,\epsmx]$, where $\epsmn$ and $\epsmx$ are inputs.
We prove the total runtime is bounded by $O(\epsmn^{-2}(1-\alpha)^{-2}(1-\rho)^{-2})$, though we believe improvements can be made to this bound. In addition to identifying the best-conductance set taken from the different approximations, the algorithm enables us to study the solution paths of PageRank, i.e. how the PPR diffusion scores change as the diffusion's accuracy varies. Hence, we call this method \ppra.

We describe a second algorithm optimized for speed (Section~\ref{sec:grid}) in finding sets of low conductance, as the exhaustive nature of our first method generates too much intermediate data for stricter values of $\eps$. Instead of computing the full solution paths, the second method searches for good-conductance sets over an approximate solution for each accuracy parameter taken from a grid of parameter values. The spacing of the accuracy parameters values on the grid is an additional input parameter. For this reason, we call the algorithm \pprg. For a log-spaced grid of values $\eps_0 > \eps_1> \cdots > \eps_N$, we locate the best-conductance set taken from a sweep over each $\eps_k$-approximation. The work required to compute the diffusions is bounded by $O(\eps_N\inv(1-\alpha)\inv)$; we show this yields a constant factor speedup over the practice of computing each diffusion separately. However, our method requires the same amount of work for performing the sweeps over each different diffusion. 



We begin by describing a modification to the PageRank linear system that will simplify our notation and the exposition of our algorithm. 

\subsection{A modified PageRank linear system for the push procedure}\label{sec:mod-linsys}
Recall that the goal is to solve the PageRank linear system~\eqref{eqn:prls} to the accuracy condition~\eqref{eqn:conv-vec} and then sort by the elements $\vx_j/d_j$. 
If we multiply Equation \eqref{eqn:prls} by $\mD\inv$, then after some manipulation we obtain
\[
(\mI - \alpha \mP^T) \mD\inv\vx = (1-\alpha)\mD\inv\vv.
\]
Note this transformation relies on $\mA$ being symmetric so that $\mP^T = (\mA\mD\inv)^T = \mD\inv\mA = \mD\inv \mP \mD$.
To avoid writing $\mD\inv$ repeatedly, we make the change of variables $\vy = (1/(1-\alpha))\mD\inv\vx$ and $\vb = \mD\inv\vv$. The modified system is then
\begin{equation}\label{eqn:prsym}
(\mI - \alpha \mP^T) \vy = \vb
\end{equation}
and we set $\vvk{x}{k} = (1-\alpha)\mD\vvk{y}{k}$.

Next we use this connection between $\vx$ and $\vy$ enables us to establish a convergence criterion for our algorithms that will guarantee 
we obtain an approximation with the kind of accuracy typically desired for methods related to the push operation, e.g. \eqref{eqn:conv-vec}. More concretely,
to guarantee $\| \mD\inv(\vx - \hvx)\|_{\infty} < \tfrac{1-\alpha}{\eps}$, it suffices to guarantee $\|\vy - \hvy\|_{\infty} < \eps$, so it suffices for our purposes to bound
the error of the system \eqref{eqn:prsym}.




The accuracy requirement has two components: nonnegativity, and error. We relate the solution to its residual as the first step toward proving both of these.
 Left-multiplying the residual vector for \eqref{eqn:prsym} by $(\mI - \alpha \mP^T)\inv$ and substituting $\vy = (\mI - \alpha \mP^T)\inv \vb$, we get
\[
\vy-\vvk{y}{k} = \left(\sum_{m=0}^{\infty} \alpha^m \left(\mP^T\right)^m \right)\vvk{r}{k},
\]
where the right-hand side replaces $(\mI - \alpha \mP^T)\inv$ with its Neumann series.
Note here that, if the right-hand side consists of all nonnegative entries, then it is guaranteed that $\vy - \vvk{y}{k} \geq 0$ holds.
Recall from Section~\ref{sec:acl-push-proc} that the residual update involved in the push procedure consists of adding nonnegative components to the residual, and so the residual \emph{must} be nonnegative.
Then, since $(1-\alpha)\vvk{y}{k} = \mD\inv\vvk{x}{k}$, this implies $\vx \geq \vvk{x}{k}$, proving one component of the accuracy criteria \eqref{eqn:conv-vec} is satisfied.

Next we bound the error in $\vy$ in terms of its residual, and then control the residual's norm.
Using the triangle inequality and sub-multiplicativity of the infinity norm allows us to bound $\|\vy - \vvk{y}{k}\|_{\infty}$, which implies~\eqref{eqn:conv-vec}, with the following
\begin{align*}
\sum_{m=0}^{\infty} \alpha^m \left\| \left(\mP^T\right)^m \vvk{r}{k} \right\|_{\infty} &\leq \left(  \sum_{m=0}^{\infty} \alpha^m \left\| \mP^T \right\|_{\infty}^m \right) \left\| \vvk{r}{k}\right\|_{\infty}.
\end{align*}
Finally, since $\mP$ is column stochastic, $\mP^T$ is row-stochastic, and so $\|\mP^T\|_{\infty} = 1$. Substituting this and noting that $\sum_{m=0}^{\infty} \alpha^m = 1/(1-\alpha)$ allows us to bound
\[
\tfrac{1}{1-\alpha}\left\| \mD\inv\vx - \mD\inv\vvk{x}{k}\right\|_{\infty} = 
\left\|\vy - \vvk{y}{k}\right\|_{\infty} \leq \tfrac{1}{1-\alpha} \left\| \vvk{r}{k}\right\|_{\infty}.
\]
So to guarantee $\vx$ satisfies the desired accuracy, it is enough to guarantee that
\begin{equation}\label{eqn:conv-crit}
\left\| \vvk{r}{k}\right\|_{\infty} < \eps
\end{equation}
holds, where $\vvk{r}{k} = \vb - (\mI - \alpha \mP^T)\vvk{y}{k}$ and $\vvk{x}{k} = (1-\alpha)\mD\vvk{y}{k}$. Thus, for our algorithms to converge to the desired accuracy, it suffices to iterate until the residual norm satisfies the bound~\eqref{eqn:conv-crit}.
With this terminating condition established,
 we can now describe our algorithm for computing the solution paths of $\vx_{\eps}$ as $\eps$ varies.

\subsection{PageRank solution paths}\label{sec:path}
\label{sec:alg-path}
Recall that our goal is computing the solution paths of seeded PageRank with respect to the parameter $\eps$. That is, we want an approximation $\vx_{\eps}$ of PageRank for all $\eps$ values inside some region. Let $\mP$ be a stochastic matrix, choose $\alpha$ satisfying $0 < \alpha < 1$, let $\vv$ be a stochastic vector, and set $\vb = \mD\inv\vv$. Fix input parameters $\epsmn$ and $\epsmx$. Then for each value $\epscu \in [\epsmn, \epsmx]$ ($\epscu$ denotes ``the value of $\eps$ currently being considered"), we want an approximation $\hvy$ of the solution to $( \mI - \alpha \mP^T ) \vy = \vb$ that satisfies $\| \vy - \hvy\|_{\infty} < \tfrac{\epscu}{1-\alpha}$. (Or rather, we want a computable approximation to this information.) As discussed in Section~\ref{sec:regularization}, we also use the approximation parameter $\rho \in [0,1)$ in the push step.

Given initial solution $\vvk{y}{0} = \textbf{0}$ and residual $\vvk{r}{0} = \vb$, proceed as follows. Maintain a priority queue, $Q(\vr)$, of all entries of the residual that do not satisfy the convergence criterion $\vr_j < \epsmn$. We store the entries of $Q(\vr)$ using a max-heap so that we can quickly determine $\|\vr\|_{\infty}$ at every step.

Each time the value $\|\vr\|_{\infty}$ reaches a new minimum, we consider the resulting solution vector to satisfy a new ``current" accuracy, which we denote $\epscu$.
For each such $\epscu$ achieved, we want to perform a sweep over the solution vector.
Because the sweep operation requires a sorted solution vector, we keep $\vy$ in a sorted array, $L(\vy)$. By re-sorting the solution vector each time a single entry $\vy_j$ is updated, we avoid having to do a full sweep for each ``new" $\epscu$-approximation. The local sorting operation is a bubblesort on a single entry; the local sweep update we describe below.

With the residual and solution vector organized in this way, we can quickly perform each step of the above iterative update. Then, iterating until $\|\vr\|_{\infty} < \epsmn$ guarantees convergence to the desired accuracy. Next we present the iteration in full detail.

\paragraph{PPR path algorithm}
The \ppath algorithm performs the following iteration until the maximum entry in $Q(\vr)$ is below the smallest parameter desired, $\epsmn$.
\begin{enumerate}
\item[ 1. ] Pop the max of $Q(\vr)$, say entry $j$ with value $r$, then set $\vr_j=\rho\epscu$ and reheap $Q(\vr)$.
\item[ 2. ] Add $r-\rho\epscu$ to $\vy_j$.
\item[ 3. ] Bubblesort entry $\vy_j$ in $L(\vy)$.
\item[ 4. ] If $L(\vy)$ changes, perform a local sweep update.
\item[ 5. ] Add $(r-\rho\epscu)\alpha\mP^T\ve_j$ to $\vr$.
\item[ 6. ] For each entry $i$ of $\vr$ that was updated, if it does not satisfy $r_i < \epsmn$, then insert (or update) that entry in $Q(\vr)$ and re-heap.
\item[ 7. ] If $\|\vr\|_{\infty} < \epscu$, record the sweep information, then set $\epscu = \|\vr\|_{\infty}$.
\end{enumerate}

When the max-heap $Q(\vr)$ is empty, this signals that all entries of $\vr$ satisfy the convergence criterion $r_j < \epsmn$, and so our diffusion score approximations satisfy the accuracy requirement \eqref{eqn:conv-vec}.

\paragraph{Sweep update}
The standard sweep operation over a solution vector involves sorting the entire solution vector and iteratively computing the conductance of each consecutive sweep set. Here, we re-sort the solution vector after each update by making only the local changes necessary to move entry $\vy_j$ to the correct ranking in $L(\vy)$. This is accomplished by bubblesorting the updated entry $\vy_j$ up the rankings in $L(\vy)$. Note that if $\vvk{y}{k}$ has $T_k$ nonzero entries, then this step can take at most $T_k$ operations. We believe this loose upperbound can be improved.
We \emph{could} determine the new rank of node $\vy_j$ in work $\log T_k$ via a binary insert. However, since we must update the rank and sweep information of each node that node $\vy_j$ surpasses, the asymptotic complexity would not change.

Once the node ranks have been corrected, the conductance score update proceeds as follows. Denote by $S^{(k-1)}(m)$ the set of nodes that have rankings $1, 2, \cdots, m$ during step $k-1$. Assuming we have the cut-set (cut and volume) information for each of these sets, then we can update that information for the sets $S^{(k)}(m)$ as follows.

Suppose the node that changed rankings was promoted from rank $j$ to rank $j-\Delta_k$.
Observe that the sets $S^{(k)}(m)$ and their cut-set information remain the same for any set $S^{(k)}(m)$ lying inside the rankings $[1, \cdots, j-\Delta_k -1]$, because the change in rankings happened entirely in the interval $[ j-\Delta_k, \cdots, j]$. This occurs for $m < j-\Delta_k$. Similarly, any set $S^{(k)}(m)$ with $m > j$ would already contain all of the nodes whose rank changed -- altering the ordering within the set does not alter the conductance of that set, and so this cut-set information also need not be changed. Hence, we need to update the cut-set information for only the intermediate sets.

Now we update the cut-set information for those intermediate sets. We refer to the node that changed rank as node $L(j)$. Its old rank was $j$, and its new rank is $j-\Delta_k$. Note that the cut-set information for the set $S^{(k)}(j-t)$ (for $t = 0, \cdots, \Delta_k$) is the exact same as that of set $ S^{(k-1)}(j-t-1)\cup \{L(j)\}$. In words, we introduce the node $L(j)$ to the set $S^{(k-1)}(j-t-1)$ from the previous iteration, and then compute the cut-set information for the new iteration's set, $S^{(k)}(j-t)$, by looking at just the neighborhood of node $L(j)$ a single time. This provides a great savings over simply reperforming the sweep procedure over the entire solution vector up to the index where the rankings changed.

If the node being operated on, $L(j)$, has degree $d$, then this process requires work $O(d + \Delta_k)$. As discussed above, we can upperbound $\Delta_k$ with the total number of iterations the algorithm performs $T_k$.

\begin{theorem}\label{thm:ppra}
Given a random walk transition matrix $\mP = \mA\mD\inv$, stochastic vector $\vv$, and input parameters $\alpha \in (0,1)$, $\rho \in [0,1)$, and $\epsmx > \epsmn >0$, our \ppra algorithm outputs the best-conductance set found from sweeps over $\epscu$-accurate degree-normalized, $\rho$-approximate solution vectors $\hvx$ to $(\mI - \alpha \mP)\vx = (1-\alpha)\vv$, for all values $\epscu \in [\epsmn,\epsmx]$. The total work required is bounded by $O\bigl(\tfrac{1}{\epsmn^2(1-\alpha)^2(1-\rho)^2}\bigr)$.
\end{theorem}

\begin{proof}
We carry out the proof in two stages. First, we show that the basic iterative update converges in work $O(\epsmn\inv(1-\alpha)\inv (1-\rho)\inv)$. Then, we show that the additional work of sorting the solution vector and sweeping is bounded by $O(\epsmn^{-2}(1-\alpha)^{-2}(1-\rho)^{-2} )$.

\textbf{Push work.} We count the work on just the residual $\vvk{r}{k}$ and solution vector $\vvk{y}{k}$. The work required to maintain the heap $Q$ and sorted array $L$ is accounted for below.

Each step, the push operation acts on a single entry in the residual that satisfies $r_j \geq \epsmn $.
The step consists of a constant number of operations to update the residual and solution vectors (namely, updating a single entry in each).
The actual amount that is removed from the residual node is $(r_j-\rho\epsmn)$; then we add $(r_j-\rho\epsmn)$ to the appropriate entry of the solution, and $(r_j-\rho\epsmn) \alpha/d_j$ to $\vvk{r}{k}_i$ for each neighbor $i$ of node $j$.
 Since $j$ has $d_j$ such neighbors, the total work in one step is bounded by $O\left(d_j\right)$.
If $T$ steps of the push operation are performed, then the amount of work required to obtain an accuracy of $\epsmn$ is bounded by $ \sum_{t=0}^T d_{j} $, where $j = j(t)$ is the index of the residual operated on in step $t$, $\vvk{r}{t}_j$.

Next we bound this expression for the work done in these ``push" steps.
Since all entries of the solution and residual vectors are nonnegative at all times, the sum of the values $(r_t-\rho\epsmn)$ pushed at each step exactly equals the sum of the values $\vvk{y}{k}$, i.e. $\sum_{t=0}^T (r_t-\rho\epsmn) = \ve^T\vvk{y}{k}$.
Since $\vvk{y}{k} = (1/(1-\alpha))\mD\inv\vvk{x}{k}$, we then have that the sum of entries in $(1/(1-\alpha))\vvk{x}{k}$ equals the sum of values pushed from the residual scaled by degree and $(1-\alpha)$, i.e.
$
\ve^T\vvk{x}{k} = (1-\alpha)\sum_{t=0}^T (r_t-\rho\epsmn)\cdot d_{j(t)},
$
 where $j(t)$ is the node pushed in step $t$. We claim that the sum $\ve^T\vvk{x}{k} \leq 1$. Assuming this for the moment, we get from the previous equation that
$
(1-\alpha)\sum_{t=0}^T (r_t-\rho\epsmn)\cdot d_{j(t)} =  \ve^T\vvk{x}{k} \leq 1.
$
Since each step of \ppath operates on a residual value satisfying $r_t \geq \epsmn$, we know that $(r_t-\rho\epsmn)\geq \epsmn(1-\rho)$, and so 
\[
(1-\alpha)\sum_{t=0}^T \epsmn(1-\rho) \cdot d_{j(t)} < (1-\alpha)\sum_{t=0}^T r_t\cdot d_{j(t)} \leq 1.
\]
Dividing by $\epsmn(1-\alpha)(1-\rho)$ completes the proof that the expression for work, $\sum_{t=0}^T d_{j(t)}$, is bounded by $O\left(\epsmn\inv(1-\alpha)\inv (1-\rho)\inv \right)$.

Lastly, we justify the claim $\ve^T\vvk{x}{k}\leq 1$.  Left-multiplying the equations in \eqref{eqn:prsym} by $(\mD\ve)^T$ and using stochasticity of $\vv$ gives
\begin{align}
\ve^T (\mI - \alpha \mP)\mD\vvk{y}{k} &= \ve^T\mD\vb - \ve^T\mD\vvk{r}{k} \nonumber\\
(1-\alpha)\ve^T\tfrac{1}{(1-\alpha)}\vvk{x}{k} &= \ve^T\vv - \ve^T\mD\vvk{r}{k} \nonumber \\
\ve^T\vvk{x}{k} &= 1 - \ve^T\mD\vvk{r}{k}.  \label{eqn:etx}
\end{align}
As noted above, all entries of the residual and iterative solution vector are nonnegative at all times. The sum $\ve^T\vvk{x}{k}$ cannot exceed 1, then, because that would imply that the residual summed to a negative number, contradicting nonnegativity of the residual vector. Hence,  $\ve^T\vvk{x}{k} \leq 1$.

\textbf{Sorting and sweeping work.}
Here we account for the work performed each step in maintaining the residual heap $Q(\vr)$,  re-sorting the solution vector $L(\vy)$, and updating the sweep information for $L(\vy)$. 
To ease the process, we first fix some notation: denote the number of entries in the residual heap $Q(\vr)$ by $|Q|$, and the number of non-zero entries in the sorted solution vector $L(\vy)$ by $|L|$. We will bound both of these quantities later on. 
 We continue to use $\Delta_t$ to denote the number of rank positions changed in $L(\vy)$ in step $t$. Finally, recall that $T$ denotes the number of iterations of the algorithm required to terminate.

The work bounds we will prove, listed in the order in which the \ppra algorithm performs them, are as follows:

\begin{tabularx}{0.9\linewidth}{XXX} \toprule
Operation & actual work & upperbound \\
\midrule
Find $\max(\vr)$                 & $1$         & $1$ \\
Delete $\max(\vr)$               & $\log(|Q|)$ & $\log( \tfrac{1}{\epsmn(1-\alpha)(1-\rho)} )$ \\
Bubblesort $L(\vy_j)$   & $\Delta_t$  & $T$ \\
Re-sweep  $L(\vy)$  & $d_j + \Delta_t$ & $d_j + T$ \\
Update $\vr + r\alpha \mP^T\ve_j $ & $d_j$      & $d_j$ \\
Re-heap  $Q(\vr)$ & $d_j \log(|Q|)$ & $d_j \log( \tfrac{1}{\epsmn(1-\alpha)(1-\rho)} )$ \\
\bottomrule
\end{tabularx}\\

The residual heap operations for deleting max $Q(\vr)$ and re-heaping the updated entries each require $O(\log(|Q|))$ work, where $|Q|$ is the size of the heap, i.e. the number of nonzero entries in the residual. We can upperbound this number using the total number of pushes performed (since a nonzero in the residual can exist only via a push operation placing it there). We bound $|Q|$ by $O(\epsmn\inv (1-\alpha)\inv (1-\rho)\inv)$, then. We remark that this is quite loose, as values of $\rho$ near 1 actually force the solution and residual to be \emph{sparser}, so the heap size should still be bounded by $O(\epsmn\inv (1-\alpha)\inv)$, though we do not yet have a proof of this.

Re-sorting the solution vector via a bubblesort can involve no more operations than the length of the solution vector. Since a nonzero in entry $\vy_j$ can exist only if a step of the algorithm operates on an entry $\vr_j$, the number of nonzeros in $\vy$ is bounded by the number of steps of the algorithm, i.e. $|L| \leq T$. We believe this bound to be loose, but cannot currently tighten it.
Note that the work required in updating sweep information also requires $\Delta_t$ work, which we again upperbound by $T$. The $d_j$ term in updating sweep information is from accessing the neighbors of the entry $\vy_j$, the node changing its rank.

The dominant terms in the above expression for work are the re-heap updates and the bubblesort and re-sweep operations, which require a total of $O(d_j \log(|Q|) + |L|)$ work each step.
Summing this over all $T$ steps of the algorithm, we can majorize work by $O(\log(|Q|)\cdot\sum_{t=0}^T d_j ) + O(\sum_{t=0}^T |L|)$, which is upperbounded by 
$O\left( \tfrac{1}{\epsmn(1-\alpha)(1-\rho)}\log(|Q|) + T\cdot |L| \right).$
Finally, substituting in our loose upperbounds for $T$, $|Q|$, and $|L|$ mentioned above completes the proof:
\[
O\left( \tfrac{1}{\epsmn(1-\alpha)(1-\rho)}\log(\tfrac{1}{\epsmn(1-\alpha)(1-\rho)}) + \tfrac{1}{\epsmn^2(1-\alpha)^2(1-\rho)^2} \right)
 \leq O\left( \tfrac{1}{\epsmn^2(1-\alpha)^2(1-\rho)^2} \right).
 \]
 
 \end{proof}

\subsection{Fast multi-parameter PPR}\label{sec:grid}


Here we present a fast framework for computing $\eps$-approximations of a push-based PPR diffusion without computing a new diffusion for each $\eps$. This enables us to identify the optimal output that would result from multiple diffusion computations for different $\eps$ values, but without having to do the work of computing a new diffusion for each different  $\eps$. This algorithmic framework does not admit the parameter $\rho$ as easily, because of implementation details surrounding the data structures used to handle sorting and updating the residual.

The framework is compatible with every set of parameter choices for $\eps$ that allows for constant-time bin look-ups. More precisely, the set of parameters $\eps_0$, $\eps_1$, $\dots$, $\eps_N$ must have an efficient method for determining the index $k$ such that, given a value $r$, we have $\eps_{k-1} > r \geq \eps_k$. We focus on a set of $\eps$ values that are taken from a log-spaced grid: that is, the parameters are of the form $\eps_k = \eps_0 \theta^k$ for constants $0 < \eps_0,\theta < 1$. Because we assume our $\eps$ parameters are taken from such a grid, we call our method \pprg.
Another possibly useful case is choosing $\eps_k$ values taken from a grid formed from Chebyshev-like nodes, allowing for constant-time shelf-placement via $\cos\inv$ evaluations.

We emphasize that the underlying algorithm we use to compute the PageRank diffusion is closely related to the push method discussed in Section~\ref{sec:alg-push} as implemented by \cite{andersen2006-local}; in the case that only a single accuracy parameter is used, the algorithms are identical. When more than one accuracy setting is used, we employ a special data structure, which we call a shelf.

\paragraph{The shelf structure}
The main difference between our algorithm \pgrid and previous implementations of the push method lies in our data structure replacing the priority queue, $Q$, discussed in \ppath . Instead of inserting residual entries in a heap as in \ppath, we organize them in a system of arrays. Each array holds entries between consecutive values of $\eps_k$, so that each array holds entries larger than the shelf below it. For this reason, we call this system of arrays a ``max-shelf", $H$, and refer to each individual array as a ``shelf", $H_k$.

The process is effectively a bucket sort: each shelf (or bucket) of $H$ holds entries of the residual lying between consecutive values of $\eps_k$ in the parameter grid.
For parameters $\eps_0, \eps_1,$ $\dots$, $\eps_N$, shelf $H_k$ holds residual values $r$ satisfying $\eps_{k-1} > r \geq \eps_{k}$. Residual entries smaller than $\eps_N$ are omitted from $H$ (since convergence does not require operating on them). Residual entries with values greater than $\eps_0$ are simply placed in shelf $H_0$.

\paragraph{PPR on a grid of $\eps$ parameters}
During the iterative step of \pgrid, then, rather than place a residual entry at the back of $Q$, we instead place the entry at the back of the appropriate shelf, $H_k$. Once all shelves $H_m(\vr)$ are cleared for $m \leq k$, then the residual has no entries larger than $\eps_k$, and so we have arrived at an approximation vector satisfying convergence criterion \eqref{eqn:conv-vec} with accuracy $\eps_k$. At this point, we perform a sweep procedure using the $\eps_k$-solution.
We then repeat the process until the next shelf is cleared, and a new $\eps_{k+1}$-solution is produced.

\textbf{PPR grid algorithm.}
The iterative step is as follows:
\begin{enumerate}
\item[ 1. ] Determine the top-most non-empty shelf, $H_k$.
\item[ 2. ] While $H$ contains an entry in shelf $k$ or above, do the following:
\item[ 3. ] \hspace*{1em} Pop an entry on or above shelf $H_k$, say value $r$ in entry $\vr_j$, and set  $\vr_j=0$.
\item[ 4. ] \hspace*{1em} Add $r$ to $\vx_j$.
\item[ 5. ] \hspace*{1em} Add $r\alpha\mP^T\ve_j$ to $\vr$.
\item[ 6. ] \hspace*{1em} For each entry of $\vr$ that was updated, move that node to the correct shelf, $H_m$, where $\eps_{m-1} > r \geq \eps_m$. If an entry is placed on a shelf higher than $k$, record the new top-shelf.
\item[ 7. ] Shelves $0$ through $k$ are cleared, so the $\eps_k$-solution is done; perform a sweep.
\end{enumerate}
Once all shelves are empty, the approximation with strictest accuracy, $\eps_N$, has been attained, and a final sweep procedure is performed.

\textbf{Shelf computation.}
In each iteration of \pprg we must place multiple entries into their respective ``shelves".
Here we show that computing the correct shelf where a value $r$ will be placed can be accomplished in constant time.

Let $\eps_k = \eps_0 \theta^k$ for a fixed value of $\theta \in (0,1)$. We want a value $r$ satisfying $\eps_{k-1} > r \geq \eps_k$ to be placed on shelf $k$. If $r \geq \eps_0$, then we place $r$ into shelf 0. Otherwise, making the substitution $\eps_k = \eps_0 \theta^k$ and performing some algebra yields
\[
k-1 < \frac{\log(r/\eps_0)}{\log(\theta)} \leq k,
\]
so $k$ can be computed by taking the ceiling of $\log(r/\eps_0)/\log(\theta)$, which is a constant time operation. Note that this process requires that $0 < \eps_k < 1$ holds for all $k$, that $\theta \in (0,1)$, and that $r > 0$.

\textbf{Top shelf.}
Each step of \pprg also requires determining the top non-empty shelf. This can be done in constant time by tracking what the top shelf is during each residual update. If $k$ is the top shelf immediately prior to step (2.4), then $k$ will still be the top shelf after the residual update is complete, unless one of the updates in step (6.) moves an entry to a shelf $l < k$. By checking for this event during the update of each individual residual entry in step (6.),  we will have knowledge of the top non-empty shelf at the beginning of each step, with only constant work per step.

Once the current working shelf is emptied, then it is possible that the next non-empty shelf is many shelves down, i.e. shelves $H_k$ and higher are emptied and the next non-empty shelf is $H_{k+c}$ for some large number $c$. Then determining $k+c$ takes $O(c)$ operations. However, this operation is performed every time the algorithm switches from one value of $\eps_k$ to the next. If there are $N$ values of $\eps_k$, then the total work in all calls of this top-shelf computation is bounded by $O(N)$.

\begin{theorem}\label{thm:pgrid}
Given a random walk transition matrix $\mP = \mA\mD\inv$, stochastic vector $\vv$, and input parameters $\alpha, \theta \in (0,1)$ and $\eps_k = \eps_0\theta^k$, our \pgrid algorithm outputs the best-conductance set found from sweeps over $\eps_k$-accurate degree-normalized solution vectors $\hvx$ to $(\mI - \alpha \mP)\vx = (1-\alpha)\vv$, for all values $\eps_k$ for $k = 0$ through $N$. The work in computing the diffusions is bounded by $O(\smash{\tfrac{1}{\eps_N(1-\alpha)}} )$. This improves on the method of computing the $N$ diffusions separately, which is bounded by $O\bigl(\tfrac{1}{\eps_N(1-\alpha)(1-\theta)} (1- \theta^{N+1}) \bigr)$. The two methods perform the same amount of sweep-cut work.
\end{theorem}

\textbf{Proof.}
Note that the amount of push-work required to produce a diffusion with smallest accuracy $\eps_N$ is exactly the same as the push-work performed in computing an $\eps_N$ solution via \ppath; The only difference is in how we organize the residual and solution vectors. Hence, the push-work for \pgrid is bounded by $O(\eps_N\inv (1-\alpha)\inv)$. Updating the shelf structure for \pgrid requires only a constant number of operations in each iteration, and so the dominating operation in one step of \pgrid is the residual push work. Thus, the push-work bound for \pgrid is $O(\eps_N \inv (1-\alpha)\inv )$.

\textbf{Push-work for $N$ separate diffusions.}
As noted above, computing a diffusion with parameters $\eps_k$ and $\alpha$ requires push-work $O( \eps_k\inv (1-\alpha)\inv)$. 
Summing this over all values of $\eps_k$ gives
$\sum_{k=0}^N \eps_k\inv (1-\alpha)\inv  = (1-\alpha)\inv\sum_{k=0}^N (1/\eps_k)$. Substituting $\eps_0\theta^k$ in place of $\eps_k$, we see this sum is simply a scaled partial geometric series, $\sum_{k=0}^N \eps_k\inv = \eps_0\inv\theta^{-N}(1-\theta^{N+1})/(1-\theta)$.
Simplifying gives 
\begin{align*}
\sum_{k=0}^N \tfrac{1}{\eps_k(1-\alpha)} &= \tfrac{1}{\eps_N (1-\alpha)(1-\theta)}\left( 1-\theta^{N+1} \right),
\end{align*}
proving the bound on the push-work.
For our choices $\eps_0 = 10^{-1}$, $\eps_N = 10^{-6}/3$, and $\theta = 0.66$ (which correponds to using $N = 32$ diffusions), this quantity is roughly 2.9 times greater than computing only one diffusion, as our method does.

\textbf{Sweep work.}
The number of operations required in computing the diffusion is bounded by $O(\eps_N \inv (1-\alpha)\inv )$, but this does not include the work done in sweeping over the various $\eps_k$-approximation vectors. The sweep operation requires sorting the solution vector. As noted in the proof of work for \ppath, the number of nonzeros in the solution vector is bounded by $O(\eps_N\inv (1-\alpha)\inv )$, and so the sorting work is $O( \eps_N\inv (1-\alpha)\inv \log(\eps_N\inv (1-\alpha)\inv) )$. This implies that sorting is the dominant subroutine of the algorithm. In practice the bound on the number of nonzeros in the solution is loose, and the push operations comprise most of the labor.


\section{Experimental Results on Finding Small Conductance Sets}
\label{sec:experiments}
We have presented two frameworks for computing a single personalized PageRank diffusion across multiple parameter settings. Here we analyze their performance on a set of real-world social and information networks with varying sizes and edge-densities with the goal of identifying sets of small conductance. All datasets were altered to be symmetric and have 0s on their diagonals; this is done by deleting any self-edges and making all directed edges undirected. In addition to versions of the Facebook dataset analyzed in Section~\ref{sec:paths}, we test our algorithms on graphs including twitter-2010 from~\cite{Kwak2010-Twitter}, friendster and youtube from~\cite{Mislove-2007-measurement,Yang-2012-ground-truth}, dblp-2010 and hollywood-2009 in~\cite{Boldi-2008-QueryFlow,Boldi-2011-layered}, idk0304 from~\cite{Caida-2005-network}, and ljournal-2008 in~\cite{Chierichetti:2009:CSN:1557019.1557049}. See Table~\ref{tab:datasets} for a summary of their properties.

\begin{table}
\centering\fontsize{8}{9}\selectfont
\caption{Datasets} \label{tab:datasets}
\begin{tabularx}{\linewidth}{XXXX}
\\
\toprule
Graph                     &    $|V|$    &       $|E|$    &    $d_{\text{ave}}$    \\
\midrule
\texttt{itdk0304}      &     190,914 &        607,610 &   6.37 \\ 
\texttt{dblp}          &     226,413 &        716,460 &   6.33  \\
\texttt{youtube}      &   1,134,890 &      2,987,624 &   5.27  \\
\texttt{fb-one}        &   1,138,557 &      4,404,989 &   3.9 \\
\texttt{fbA}              &   3,097,165 &     23,667,394 &  15.3 \\%
\texttt{ljournal}         &   5,363,260 &     49,514,271 &  18.5 \\
\texttt{hollywood}        &   1,139,905 &     56,375,711 &  98.9 \\
\texttt{twitter}          &  41,652,230 &  2,041,892,992 & 98 \\
\texttt{friendster}       &  65,608,366 &  1,806,067,135 & 55.1 \\%
 \bottomrule
\end{tabularx}
\end{table}


\subsection{The effect of $\rho$ on conductance}\label{sec:rho-scaling}

Our first experimental study regards the selection of the parameter $\rho$ for finding sets of small conductance. We already established that $\rho = 0.9$ yielded qualitatively accurate solution path plots. However, for the specific problem of identifying small conductance sets, we find a curious behavior and get the best results with small values of $\rho$. We'll explain why this is shortly, but consider the results in Figure~\ref{fig:rho-scaling}. In the left subplot, we see the maximum difference between the minimum conductance found for any value of $\rho$ over a series of trials. It can be large, for instance, $0.7$ for one trial on the LiveJournal graph, where large $\rho$ shows worse results. In that same figure, we show the runtime scaling. It seems to scale with $1/(1-\rho)$, which is slightly better than expected from the bound in Theorem~\ref{thm:ppra}.

\begin{figure*}[h]
\centering
\includegraphics[width=0.5\linewidth]{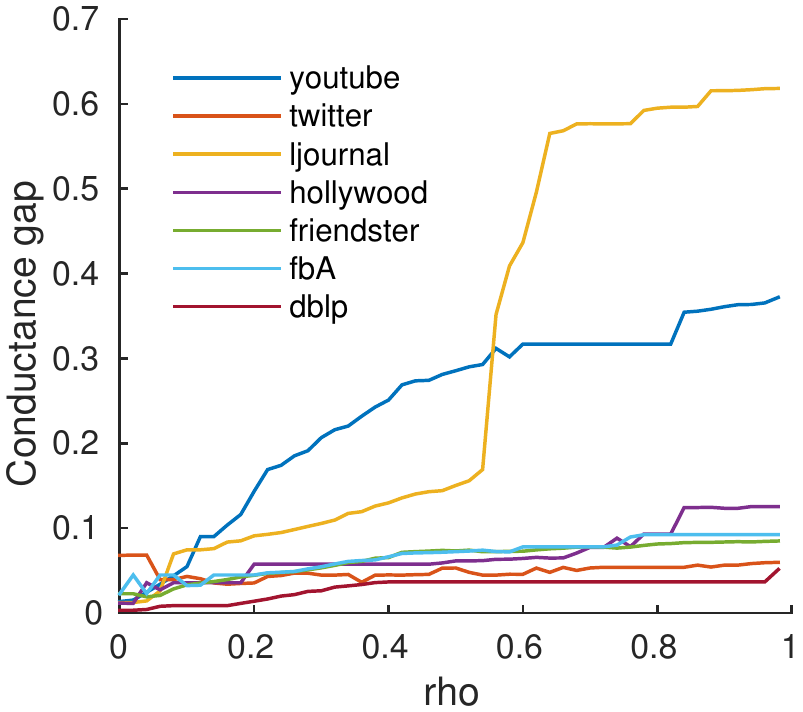}%
\includegraphics[width=0.5\linewidth]{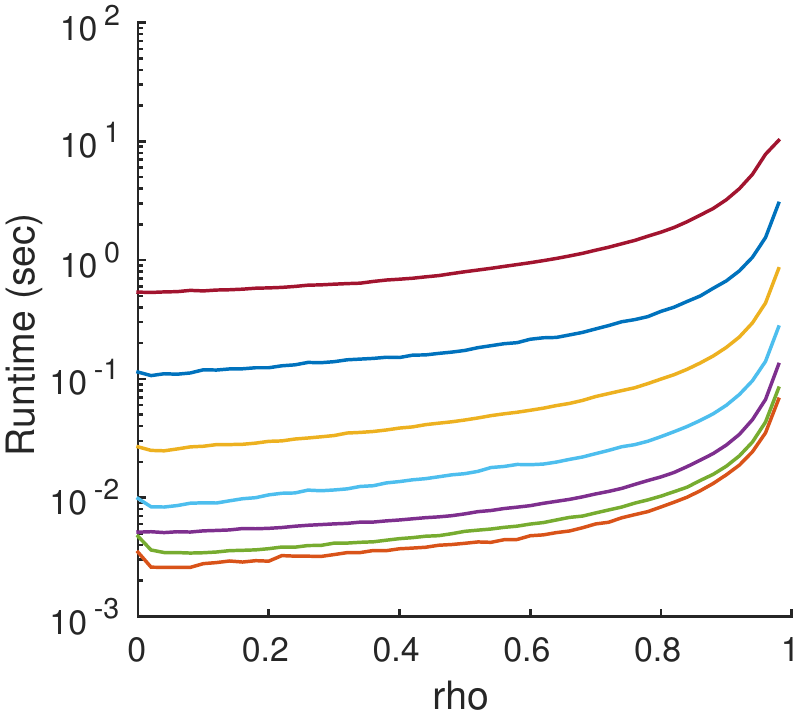}
\caption{
Here we display the behavior of the solution paths as $\rho$ scales from 0 to 1.
At left, we display the gap between $\phi(\rho)$, the best conductance found at that value of $\rho$, and $\phi_{\min}$, the minimum conductance found over all values of $\rho$. The lines depict the maximum difference over 100 trials of the quantity $\phi(\rho) - \phi_{\min}$. This plot shows that
the best conductance found becomes \emph{worse} as $\rho$ approaches 1. At right, the runtime appears to scale with $1/(1-\rho)$, which is better than the $1/(1-\rho)^2$ predicted by our theory.
}
\label{fig:rho-scaling}
\end{figure*}

The greatest difference between the best conductance found for any value of $\rho$ and the worst conductance found for any $\rho$ occurs in the livejournal graph, with a gap of nearly $0.7$. We discovered that the cause for this disparity is that large values of $\rho$ delay the propagation of the diffusion, and so 
the $\rho = 0.9$ paths at $\eps = 10^{-5}$ did not spread far enough to find a set of conductance near $0.07$. In contrast, all paths with $\rho < 0.5$ did diffuse deep enough into the graph to identify this good conductance set. Thus, it is possible that many of the differences in conductance performance between paths with different values of $\rho$ might in fact be caused by the \emph{size} of the region to which the diffusion spreads for a given value of $\eps$. Figure~\ref{fig:lj-anomaly} illustrates this finding.

\begin{figure*}[h]
\centering
\includegraphics[width=0.5\linewidth]{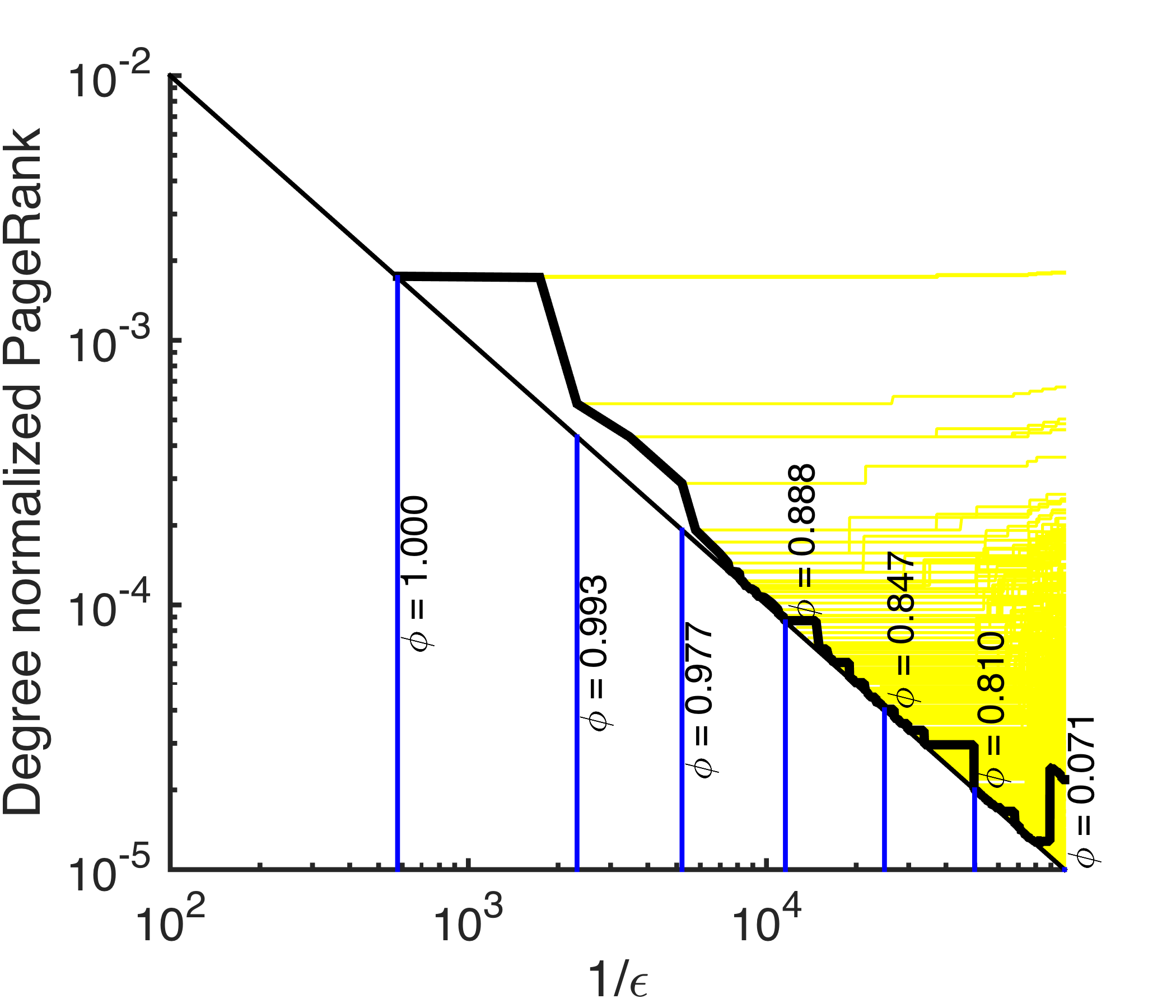}%
\includegraphics[width=0.5\linewidth]{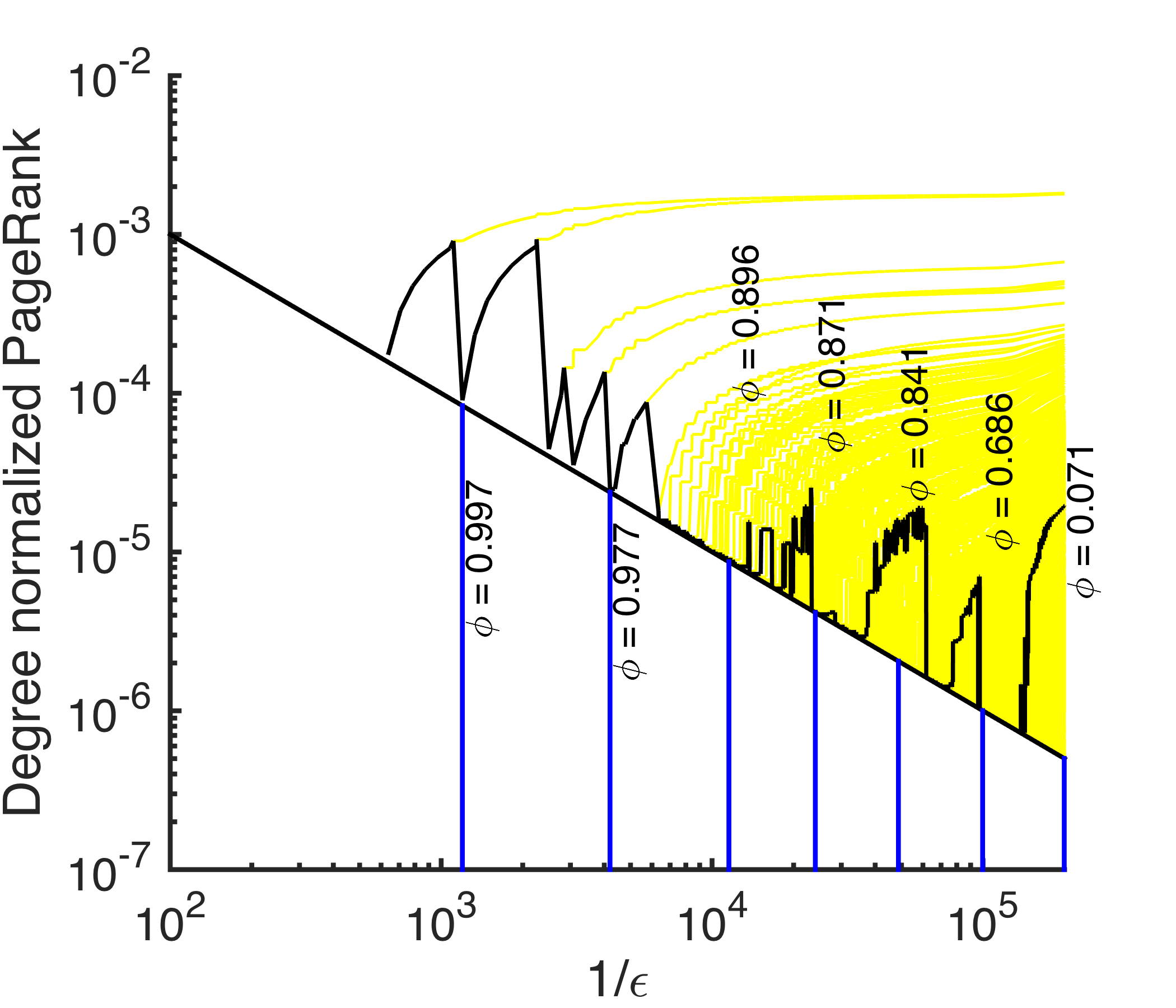}%
\caption{
At left, the $\rho = 0$ paths identify mostly poor conductance sets $\phi \approx 0.8$, and locate a set of low conductance, $\phi = 0.0788$, only toward the end of the diffusion. At right we see that the $\rho=0.9$ paths cannot find this set with $\eps = 10^{-5}$. With a slightly smaller accuracy ($\eps = 5\cdot10^{-6}$ instead of $\eps = 10^{-5}$), the diffusion is able to spread far enough to locate the good conductance set.
}
\label{fig:lj-anomaly}
\end{figure*}

Our conclusion from these experiments is that, for the goal of finding sets of small conductance, we should use small values of $\rho$ near zero. While it sometimes happens that $\rho > 0$ slightly improves conductance, this is not a reliable observation, and so for the remaining experiments on conductance, we set $\rho = 0$. (This has the helpful side effect of making it easier to compare with our \pgrid.)

\subsection{Runtime and conductance: \ppath}

Our first method, \ppath, is aimed at studying how PPR diffusions vary with the parameter $\eps$. Toward this, Table~\ref{tab:path-time} emphasizes the shear volume of distinct $\eps$-approximations that \ppath explores. We also want to highlight both the efficiency of our method over the na\"{\i}ve approach for computing the solution paths, and the additional information that the solution paths provide compared to a single diffusion.

With this in mind, our experiment proceeds as follows. On each data set, we selected 100 distinct nodes uniformly at random, and ran three personalized PageRank algorithms from that node,
with the settings $\alpha = 0.99$ and $\eps = 10^{-5}$. Table~\ref{tab:path-time} displays results for our solution paths algorithm (``path" in the table) compared with two other algorithms chosen to emphasize the runtime and the performance of \ppath.

To show how \ppath scales compared to the runtime of a single diffusion, and to emphasize that the solution paths can locate better conductance sets in some cases, we compare our solution paths method with a standard implementation for computing a single PPR diffusion (``single" in Table~\ref{tab:path-time}). Column 3 in the table gives the median runtime, taken over 100 trials, of the single diffusion. To compare, column 4 gives the median \emph{ratio} of ``path" time to ``single" time. Although \ppath is slower on the small graphs, on the larger graphs we see the runtime is nearly the same as for a single PPR diffusion. At the same time, column 2 shows that ``path" computes the results from hundreds or even thousands of diffusions, a significant gain in information over the single PPR diffusion. Finally, column 7 gives the best ratio of conductance found by ``path" compared to that found by ``single". This shows that the solution paths can improve conductance by $10\%$ to even $50\%$ compared to a single diffusion.

To display the efficiency of our algorithm in computing these many diffusion settings, we again use the standard PPR implementation, but this time set to compute the  diffusion for every accuracy setting $k^{-1}$ for $k = 1$ to $10,000$. This algorithm is ``mult" in Table~\ref{tab:path-time}, and is essentially a na\"{\i}ve method for approximating the solution paths. Column 5 gives the ratio of ``mult" time to ``single" time, and shows that this na\"{\i}ve approach to computing diffusions with multiple accuracies is prohibitively slow -- it is thousands of times slower than our ``path" method.

Lastly, we acknowledge here that both variations on the PPR diffusion are na\"{\i}ve approaches to the problem at hand. However, currently there is no other algorithm for computing the PPR solution paths which we can use as a more competitive baseline.


\begin{table}
\centering\fontsize{8}{9}\selectfont
\begin{tabularx}{\linewidth}{@{\hspace{0pt}}l@{\hspace{3pt}}l@{\hspace{6pt}}XXXXXXXXX@{\quad}X} \toprule
 Data & num $\eps$ 
& \multicolumn{3}{>{\hspace{-6pt}}l}{Single diff.~time (sec.)} 
& \multicolumn{3}{l}{\ppath~time (sec.)}
& \multicolumn{3}{l}{multi diff.~time (sec.)}& $\phi$-ratio\\ 
 \cmidrule(r{12pt}){3-5}
  \cmidrule(r{12pt}){6-8}
 \cmidrule(r{14pt}){9-11}
\cmidrule(r{8pt}){12-12}
   &                         & 25 & 50 & 75 & 25& 50 & 75 & 25& 50 & 75 &  \\ 
   \midrule
  itdk0304  & 5292  &   0.02 &   0.02 &   0.03  &   0.28 &   0.41 &   0.69 & 70.8 & 94.2 & 123.2 &   1.77 \\ 
      dblp  & 8138  &   0.02 &   0.02 &   0.02  &   0.40 &   0.51 &   0.65 & 87.3 & 97.9 & 111.5 &   1.12 \\ 
   youtube  & 2844  &   0.01 &   0.01 &   0.01  &   0.05 &   0.10 &   0.15 & 28.6 & 38.7 & 49.2 &   1.47 \\ 
    fb-one  & 3464  &   0.01 &   0.01 &   0.01  &   0.03 &   0.05 &   0.07 & 28.1 & 34.6 & 40.5 &   1.09 \\ 
       fbA  & 862  &   $<0.01$ &   $<0.01$ &   0.01  &   0.01 &   0.01 &   0.01 & 14.0 & 16.5 & 19.5 &   1.16 \\ 
  ljournal  & 2799  &   0.01 &   0.01 &   0.01  &   0.01 &   0.02 &   0.05 & 24.5 & 30.9 & 43.6 &   2.09 \\ 
 hollywood  & 423  &   $<0.01$ &   $<0.01$ &   $<0.01$  &   $<0.01$ &   $<0.01$ &   0.01 & 14.0 & 17.2 & 22.4 &   1.19 \\ 
   twitter  & 172  &   $<0.01$ &   $<0.01$ &   $<0.01$  &   $<0.01$ &   $<0.01$ &   0.01 & 6.5 & 10.3 & 18.1 &   1.05 \\ 
friendster  & 402  &   $<0.01$ &   $<0.01$ &   $<0.01$  &   $<0.01$ &   $<0.01$ &   0.01 & 11.1 & 13.6 & 16.6 &   1.09 \\ 
 \bottomrule
\end{tabularx}
\caption{
Runtime and conductance comparison of the solution paths (all accuracies from $10^{-1}$ to $10^{-5}$) with (1) a single PPR diffusion with accuracy $10^{-5}$ (labelled ``single'') and (2) 10,000 PPR diffusions, accuracies $k^{-1}$ for $k=1$ to 10,000 (labelled ``mult''). 
On each dataset we selected 100 distinct nodes uniformly at random and ran the algorithms with the settings $\alpha = 0.99$ and $\eps = 10^{-5}$ and $\rho = 0$.
Column ``num $\eps$'' displays the median number of distinct accuracy parameters $\eps$ explored by our algorithm \ppath.
Columns under ``Time'' report 25th, 50th, and 75th percentile of runtimes over these 100 trials. The column ``$\phi$-ratio" lists the largest (best) ratio of conductance achieved by a single diffusion with conductance achieved by our \ppath, showing our method can improve on the conductance found by a single diffusion by as much as a factor of $2.09$.
}\label{tab:path-time}
\end{table}

\subsection{Runtime and Conductance: \pgrid}
\label{sec:runtime-grid}
We compare our second method \pgrid with a method called \pgrow, which uses the push framework described in Section~\ref{sec:alg-push}. Each of these algorithms uses a variety of accuracy settings, and returns the set of best conductance found from performing a sweep-cut over the diffusion vector resulting from each accuracy setting. The algorithm \pgrow has 32 pre-set accuracy parameters $\eps_k$. In contrast with \pgrid, which takes its accuracy parameters from a log-spaced grid $\eps_k = \eps_0\theta^k$, the parameters for \pgrow are chosen as the inverses of values from the grid $10^j\cdot\left[\begin{array}{cccccc} 2 & 3 & 4 & 5 & 10 & 15 \end{array}\right]$ for $j=0,1, \cdots ,4$, along with two additional parameters, $10^{-6}/2$ and $10^{-6}/3$.

In addition to $\alpha$,	our method \pgrid has the parameters $\eps_0$ and $\eps_N$, the laxest and strictest accuracies (respectively), and $\theta$, which determines the fineness of the grid of accuracy parameters.
We use the values $\eps_0 = 10^{-1}$ and $\eps_N = 10^{-6}/3$, and use values of $\theta$ corresponding to $N = 32$, $64$, and $1256$ different accuracy parameters.

We emphasize that this comparison with the \pgrow method is not as na\"{\i}ve as it might seem: out of the 32 calls that it makes, in practice the very last call (with the strictest value of $\eps$) constitutes near $37\%$ of the total runtime. This means that making  only a single call would save little work, and would sacrifice the information from the other 31 (smaller) approximations. Furthermore, the primary optimizations that would be made to the \pgrow framework to improve on this are exactly the optimizations that we make with our \pgrid algorithm, namely avoiding re-doing push work between diffusion computations for different values of $\eps$.

Because the two algorithms compute the same PageRank diffusion, comparing their runtimes here allows us to study what proportion of the total work is  made up of redundant push operations, and what proportion is comprised of the sweep cut procedures, which both algorithms perform anew for each diffusion. To study this, we highlight the results in Table~\ref{tab:grid-v-grow-time} which displays the runtimes for \pgrow and the \emph{ratios} of the runtimes of \pgrid with \pgrow for computing the best-conductance set from the same number of different diffusions, $N=32$. We also display \pgrid results for the cases $N=64$ and $1256$ to show how the algorithm scales with the fineness of the grid.

To compare runtimes, we perform the following for each different dataset. For 100 distinct nodes selected uniformly at random, we ran both algorithms with the setting $\alpha = 0.99$. We display the best ($25\%$) and worst ($75\%$) quartile of performance of each algorithm and parameter setting.
On almost all datasets, we see that \pgrid with $N=32$ has a speedup of a factor 2 to 3.
This is consistent with our theoretical comparison of the two runtimes in Theorem~\ref{thm:pgrid}, which predicts a factor of 2.9 difference in the push-work that the two algorithms perform.
Then, columns 6 through 9 of Table~\ref{tab:grid-v-grow-time} display how quickly \pgrid can compute even more diffusions: whereas \pgrow takes around 1 second to compute and analyze $N=32$ diffusions, \pgrid takes little more than half that time to compute on $N=64$ diffusions (columns 6 and 7). Columns 8 and 9 show that \pgrid can compute and analyze $N=1256$ diffusions, nearly 40 times as many as \pgrow, in an amount of time only 1.10 to 6.59 times greater than the time required by \pgrow.


\begin{table}
\centering\fontsize{8}{9}\selectfont
\begin{tabularx}{\linewidth}{@{}lXX@{\qquad}XX@{\qquad}XX@{\qquad}XX@{}} \toprule
 & \multicolumn{2}{>{}l}{time (sec.)} & \multicolumn{2}{>{\hspace{-6pt}}l}{time ratio} & \multicolumn{2}{>{\hspace{-6pt}}l}{time ratio} & \multicolumn{2}{>{\hspace{-6pt}}l}{time ratio} \\
Data & \multicolumn{2}{>{}l}{\pgrow } & \multicolumn{2}{>{\hspace{-6pt}}l}{\pgrid$N=32$} & \multicolumn{2}{>{\hspace{-6pt}}l}{\pgrid $N = 64$} & \multicolumn{2}{>{\hspace{-6pt}}l}{\pgrid $N = 1256$} \\
\cmidrule(l{3pt}r{32pt}){2-3} %
\cmidrule(r{26pt}){4-5}
\cmidrule(r{26pt}){6-7}
\cmidrule(r{4pt}){8-9}
& 25 & 75 &
25 & 75 & 
 25 & 75 &
 25 & 75 \\
\midrule
   itdk0304  &  6.23 &  8.73  &  0.56  &  0.61  &  0.61  &  0.66  &  1.10  &  1.20  \\ 
       dblp  &  4.52 &  7.21  &  0.56  &  0.62  &  0.62  &  0.67  &  1.28  &  1.43  \\ 
    youtube  &  1.73 &  2.39  &  0.39  &  0.50  &  0.54  &  0.65  &  3.35  &  4.38  \\ 
     fb-one  &  1.25 &  1.60  &  0.33  &  0.39  &  0.45  &  0.53  &  3.72  &  4.38  \\ 
        fbA  &  0.49 &  0.65  &  0.47  &  0.55  &  0.63  &  0.72  &  5.99  &  6.59  \\ 
   ljournal  &  0.82 &  1.20  &  0.44  &  0.55  &  0.58  &  0.74  &  4.57  &  6.12  \\ 
  hollywood  &  0.28 &  0.64  &  0.34  &  0.49  &  0.44  &  0.60  &  3.47  &  5.00  \\ 
    twitter  &  0.13 &  0.37  &  0.39  &  0.44  &  0.54  &  0.60  &  4.61  &  5.44  \\ 
 friendster  &  0.34 &  0.49  &  0.39  &  0.44  &  0.51  &  0.58  &  3.90  &  4.32  \\ 
 \bottomrule
 \vspace{1pt}
\end{tabularx}
\caption{Runtime comparison of our \pgrid with \pgrow.
For each dataset, we selected 100 distinct nodes uniformly at random and ran \pgrow with 32 and \pgrid with $N$ different accuracy settings $\eps_k$. Columns 2 and 3 display the $25$th and $75$th percentile runtimes for \pgrow (in seconds). The other columns display the median over the 100 trials of the \emph{ratios} of the runtimes of \pgrid (using the indicated parameter setting) with the runtime of \pgrow on the same node. These results demonstrate 
that our algorithm computing over $N=32$ accuracy parameters $\eps_k$ achieves the factor of 2 to 3 speed-up predicted by our theory in Section~\ref{sec:grid}.
}\label{tab:grid-v-grow-time}
\end{table}

The conductances displayed in Table~\ref{tab:grid-v-grow-cond} are taken from the same trials as the runtime information in Table~\ref{tab:grid-v-grow-time}. As with the table of runtimes, for each dataset the table gives the $25\%$ (best) and $75\%$ (worst) percentiles of conductance scores produced by each algorithm on the 100 trials. We see nearly identical conductance scores for \pgrow and \pgrid with $N = 32$, which we expect because the two perform nearly identical work. It is interesting to note, however, that increasing the number of diffusions can result in significantly improved conductance scores in some cases, as with $N=1256$ on the ``fb-one" and ``hollywood" datasets. This demonstrates concretely the potential effect of using a broad swath of parameter settings for $\eps$ to study the meso-scale structure.
Moreover, it demonstrates that even a finely spaced mesh of $\eps$ values, as with \pgrow and \pgrid with $N=64$, can miss informative diffusions.

\noindent
\begin{table}
\centering\fontsize{8}{9}\selectfont
\begin{tabularx}{\linewidth}{@{}lX@{\qquad}XX@{\qquad}XX@{\qquad}XX@{}} \toprule
Data & \texttt{grow} & \multicolumn{2}{>{\hspace{-6pt}}l}{$N=32$} & \multicolumn{2}{>{\hspace{-6pt}}l}{$N = 64$} & \multicolumn{2}{>{\hspace{-6pt}}l}{$N = 1256$} \\
\cmidrule(r{36pt}){3-4}
\cmidrule(r{36pt}){5-6}
\cmidrule(r{18pt}){7-8}
&  & 25 & 75 & 
 {25} & {75} &
 {25} & {75} \\
\midrule
   itdk0304  &  0.06  &  1.00  &  1.00  &  1.00  &  1.01  &  1.00  &  1.02  \\ 
       dblp  &  0.07  &  1.00  &  1.00  &  1.00  &  1.00  &  1.00  &  1.01  \\ 
    youtube  &  0.18  &  1.01  &  1.30  &  1.09  &  1.50  &  1.21  &  1.72  \\ 
     fb-one  &  0.37  &  1.06  &  1.16  &  1.10  &  1.26  &  1.18  &  1.37  \\ 
        fbA  &  0.56  &  1.00  &  1.05  &  1.00  &  1.06  &  1.00  &  1.09  \\ 
   ljournal  &  0.32  &  1.00  &  1.01  &  1.00  &  1.01  &  1.00  &  1.01  \\ 
  hollywood  &  0.29  &  1.00  &  1.01  &  1.00  &  1.01  &  1.00  &  1.02  \\ 
    twitter  &  0.80  &  1.00  &  1.00  &  1.00  &  1.00  &  1.00  &  1.00  \\ 
 friendster  &  0.85  &  1.00  &  1.00  &  1.00  &  1.00  &  1.00  &  1.01  \\ 
 \bottomrule
\vspace{1pt}
\end{tabularx}
\caption{Conductance comparison of our \pgrid with \pgrow.
Column 2 displays the median of the conductances found by \pgrow in the same 100 trials presented in Table~\ref{tab:grid-v-grow-time}.
The other columns display the $25\%$ and $75\%$ percentiles of the \emph{ratio} of the conductances achieved by \pgrow and \pgrid for the same seed set. For example, on the dataset `fb-one', the conductances found by \pgrow are $18\%$ larger than those found by \pgrid with $N=1256$ accuracy settings~--- and that comparison is on the quartile of trials where \pgrid compares the \emph{worst} to \pgrow. We report the ratios in this manner (rather than their reciprocals) because in this form the values displayed are greater than 1, which distinguishes the values from conductance scores (which are between 0 and 1).
} \label{tab:grid-v-grow-cond}
\end{table}

\section{Related work}
\label{sec:related}

As we already mentioned, regularization paths are common in statistics~\cite{Efron-2004-lars,Hastie-2009-elements}, and they help guide model selection questions.
In terms of clustering and community detection, solution paths are extremely important for a new type of convex clustering objective function~\cite{Hocking-2011-clusterpath,Lindsten-2011-convex-kmeans}. Here, the solution path is closely related to the number and size of clusters in the model.

One of the features of the solution path that we utilize to understand the behavior of the diffusion is the stability of the set of best conductance over time. In ref.~\cite{Delvenne2010-stability}, the authors use a closely related concept to study the persistence of communities as a different type of temporal relaxation parameter varies. Again, they use the stability of communities over regions of this parameter space to indicate high-quality clustering solutions. 

In terms of PageRank, there is a variety of work that considers the PageRank vector as a function of the teleportation parameter $\alpha$~\cite{boldi2009-functional,langville2006-book}. Much of this work seeks to understand the sensitivity of the problem with respect to $\alpha$. For instance, we can compute the derivative of the PageRank vector with respect to $\alpha$. It is also used to extrapolate solutions to accelerate PageRank methods~\cite{brezinski2005-extrapolation}. More recently, varying $\alpha$ was used to show a relationship between personalized-PageRank-like vectors and spectral clustering~\cite{Mahoney-2012-local}. Note that PageRank solution paths as $\alpha$ varies would be an equally interesting parameter regime to analyze. The parameter $\alpha$ functions akin to $\eps$ in that large values of $\alpha$ cause the diffusion to propagate further in the graph. 


\section{Conclusions and discussion}
\label{sec:conclusion}

We proposed two algorithms that utilize the push step in new ways to generate refined insights on the behavior of diffusions in networks. The first is a method to rapidly estimate the degree-normalized PageRank solution path as a function of the tolerance $\eps$. This method is slower than estimating the solution of a single diffusion in absolute run time, but still fast enough for use on large graphs. We designed that method, and the associated degree-normalized PageRank solution path plot, in order to reveal new insights about regions at different size-scales in large networks. The second method is a fast approximation to the solution path on a grid of logarithmically-spaced $\eps$ values. It uses an interesting application of bucket sort to efficiently manage these diffusions. We demonstrate that both of these algorithms are fast and local on large networks.

The seeded PageRank solution plots, in particular, are effective at identifying a number of subtle structures that emerge as a diffusion propagates from a set of seed nodes to the remainder of the network. We hope that these become useful tools to diagnose and study the properties of large networks. 

As recently established by Ghosh et al.~\cite{Ghosh-2014-cheeger}, there are many related diffusion methods that all share Cheeger-like inequalities for specific definitions of conductance. We anticipate that our solution path algorithm could apply to any of these diffusions as well. For instance, our recent result on estimating the heat kernel diffusion in large graphs is based on the push step as well~\cite{Kloster-2014-hkrelax}; we anticipate only mild difficulty in adapting our results to that diffusion.

Fast access to the solution path trajectories provides a number of additional opportunities that we have not yet explored. We may be able to track multiple clusters directly by managing intermediate data.  We may be able to find near-optimal conductance sets that are larger than those that directly optimize the objective.  Also, nodes in an egonet or larger set could be further clustered by properties of their solution paths instead of their connectivity patterns.

%
%
%
%

\section*{Acknowledgments}
We thank the following people for their careful reading of several early drafts: Huda Nassar, Bryan Rainey, and Varun Vasudevan. This work was supported by NSF CAREER Award CCF-1149756.

\bibliographystyle{abbrv}
\bibliography{main}

\label{lastpage}

\end{document}